\documentclass[10pt,twocolumn,letterpaper]{article}

\usepackage{cvpr}              

\usepackage[dvipsnames]{xcolor}

\definecolor{cvprblue}{rgb}{0.21,0.49,0.74}
\usepackage[pagebackref,breaklinks,colorlinks,allcolors=cvprblue]{hyperref}

\usepackage{algorithm}
\usepackage{algpseudocode}

\usepackage{makecell}
\usepackage{graphicx}
\usepackage{color}
\usepackage{amsfonts}
\usepackage{amsmath}
\usepackage{amssymb}

\usepackage{bm}
\usepackage{hhline}

\usepackage{amsthm}
\usepackage{aliascnt}

\newcommand{\maketheorem}[3][theorem]{
        \expandafter\let\csname #2\endcsname\undefined
        \newaliascnt{#2}{#1}
        \newtheorem{#2}[#2]{#3}
        \aliascntresetthe{#2}
        \expandafter\def\csname#2name\endcsname~##1\null{#3~##1\null}
}

\theoremstyle{definition}

\maketheorem{proposition}{Proposition}
\maketheorem{lemma}{Lemma}
\theoremstyle{definition}
\maketheorem{assumption}{Assumption}
\theoremstyle{remark}
\maketheorem{remark}{Remark}

\usepackage{multirow}
\usepackage{multicol} 
\usepackage{booktabs}

\usepackage{colortbl} 
\usepackage{caption} 
\makeatletter
\def\thanks#1{\protected@xdef\@thanks{\@thanks
        \protect\footnotetext{#1}}}
\makeatother

\algblock[Block]{Block}{EndBlock}
\algblockdefx[Block]{Block}{EndBlock}%
[1]{{#1}}%
[1]{{#1}}

\def\paperID{16611} 
\def\confName{CVPR}
\def\confYear{2025}

\definecolor{lightgreen}{RGB}{84,192,84}

\title{
    FiRe: Fixed-points of Restoration Priors for Solving Inverse Problems
}

\author{Matthieu Terris$^*$, Ulugbek S. Kamilov$^\dagger$, Thomas Moreau$^*$
\\ 
$^*$Université Paris-Saclay, Inria, CEA, Palaiseau, 91120, France\\
$^\dagger$Washington University in St.~Louis, St.~Louis, MO, USA\\
}

\begin{document}
\maketitle

\begin{abstract}
Selecting an appropriate prior to compensate for information loss due to the measurement operator is a fundamental challenge in imaging inverse problems. Implicit priors based on denoising neural networks have become central to widely-used frameworks such as Plug-and-Play (PnP) algorithms. In this work, we introduce Fixed-points of Restoration (FiRe) priors as a new framework for expanding the notion of priors in PnP to general restoration models beyond traditional denoising models. The key insight behind FiRe is that smooth images emerge as fixed points of the composition of a degradation operator with the corresponding restoration model. This enables us to derive an explicit formula for our implicit prior by quantifying invariance of images under this composite operation. Adopting this fixed-point perspective, we show how various restoration networks can effectively serve as priors for solving inverse problems. The FiRe framework further enables ensemble-like combinations of multiple restoration models as well as acquisition-informed restoration networks, all within a unified optimization approach. Experimental results validate the effectiveness of FiRe across various inverse problems, establishing a new paradigm for incorporating pretrained restoration models into PnP-like algorithms. Code available at \href{https://github.com/matthieutrs/fire}{https://github.com/matthieutrs/fire}.
\end{abstract}

\section{Introduction}

Image restoration aims to recover high-quality images from degraded measurements, encompassing fundamental low-level vision tasks such as denoising, super-resolution, and inpainting. These problems can be formulated as linear inverse problems, where one seeks to recover an image
$x\in\mathbb{R}^n$ from given observations $y\in\mathbb R^m$ obtained through
\begin{equation}
    y = Ax + e\enspace,
\label{eq:inv_pb}
\end{equation}
where $A$ is a linear operator, modeling the sensing system, and $e\in\mathbb{R}^m$ is some additive noise.
A standard approach for solving such problem consists in finding an estimate $x$ that maximizes the posterior distribution $p(x|y)$. 
Assuming that the noise $e$ is Gaussian, the Bayes' rule yields the following equivalent minimization problem:
\begin{equation}
\label{eq:bayes}
\underset{x}{\operatorname{argmin}}\,\, \frac{1}{2}\|Ax-y\|_2^2 - \log p(x),
\end{equation}
where $p$ is a \emph{prior} for the distribution of natural images.
In this context, a powerful surrogate prior $p$ is needed to ensure that a solution to \eqref{eq:bayes} is satisfying in the sense of being both visually appealing and faithful to the measurements \eqref{eq:inv_pb}.
In the modern literature, denoising neural networks have imposed themselves as powerful implicit priors.
Their popularity stems from their ability to approximate the gradient of a log prior $\nabla_x \log p(x)$, also known as the score, through the Tweedie's formula \cite{vincent2011connection, xu2020provable, laumont2022bayesian}.
Such models can be plugged within iterative optimization algorithms for solving \eqref{eq:bayes}, giving rise to the so-called Plug-and-Play (PnP) framework \cite{venkatakrishnan2013plug, pesquet2021learning, hurault2022proximal, kamilov2023pnp}.
Note that such score-matching methods are also at the heart of diffusion approaches \cite{chung2022improving, delbracio2023inversion, chung2024direct} for solving inverse problems.

From a practical perspective, PnP methods alternate between data-fidelity steps enforcing consistency with \eqref{eq:inv_pb} and denoising steps that implicitly encode natural image priors. Despite theoretical convergence results ensuring stable reconstructions \cite{xu2020provable, pesquet2021learning, hurault2022proximal, terris2023equivariant}, achieving state-of-the-art results relies heavily on heuristics such as decreasing step sizes or early stopping \cite{wang2017parameter, wei2020tuning, zhang2021plug}. This gap between theory and practice stems from a critical limitation of PnP algorithms: at convergence, denoisers trained on Gaussian-corrupted images are applied to smooth images, creating a significant distribution mismatch.

\begin{figure*}
\centering
\includegraphics[width=0.9\textwidth]{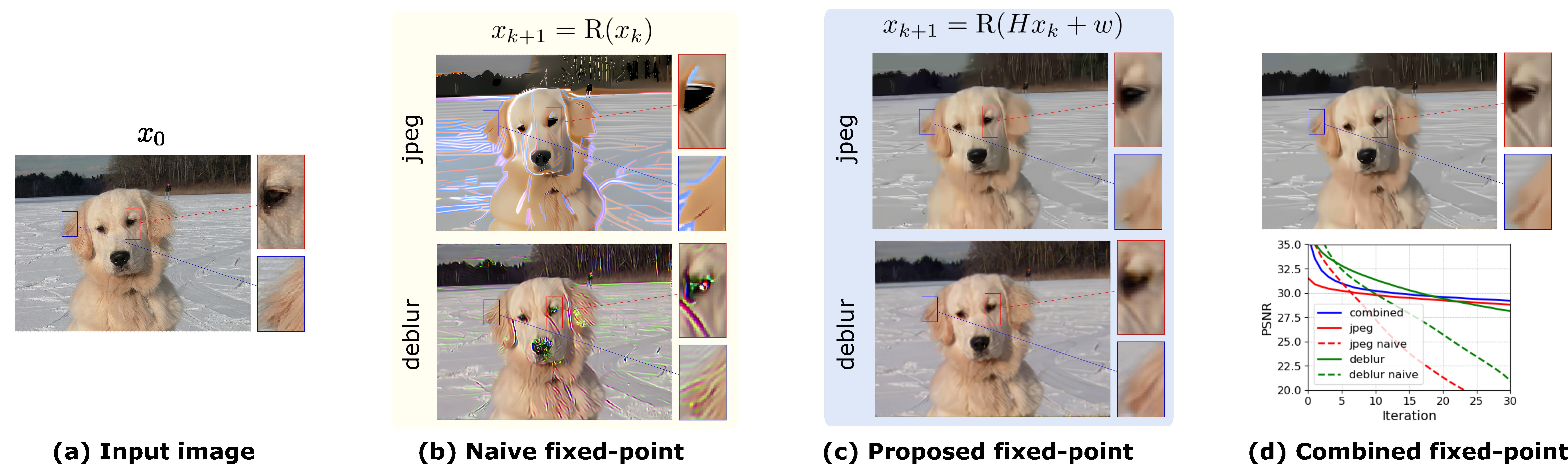}
\caption{(a) \textbf{Experimental setup.} Two restoration models trained for different tasks: $\operatorname{R}_1$ is the SCUNet from \cite{zhang2023practical} for JPEG restoration, and $\operatorname{R}_2$ is a Restormer model \cite{zamir2022restormer} fine-tuned for Gaussian deblurring. (b) \textbf{Fixed points of restoration models are not smooth images}. Direct iteration of these models in a standard PnP framework yields strong artifacts. (c) \textbf{Composition with training degradation preserves image characteristics.} For each model $\operatorname{R}_i$ and its associated degradation $D_i = H_i\cdot + w_i$, the fixed-points of $\operatorname{R}_i\circ D_i$ maintain natural image properties. (d) \textbf{Linear combinations yield improved fixed points.} Combining $\operatorname{R}_1 \circ D_1$ and $\operatorname{R}_2 \circ D_2$ produces higher quality fixed-points (top), with stable PSNR values (bottom graph, solid lines) compared to the rapidly degrading naive approach from (b) (bottom graph, dashed lines).} 
\label{fig:summary}
\end{figure*}

A natural extension of this framework is to consider other restoration models as potential implicit priors.
Recent works demonstrate that super-resolution models can be incorporated into both PnP restoration \cite{hu2023restoration} and diffusion models \cite{bansal2024cold}.
This approach is particularly appealing as it can leverage the vast collection of pretrained image restoration models available to the community.
However, this approach is not as straightforward as substituting a denoiser with a more general restoration model, as this naive approach leads to poor reconstruction quality.
Moreover, extending Tweedie's formula to general restoration models beyond denoisers is not trivial.
 
To leverage the implicit prior of restoration models in the wild, our main observation is that both denoisers and restoration models are trained to have smooth images as fixed points of the composition between the train-time image degradation (\emph{e.g.}~additive noise or linear degradation) and the restoration model. Based on this observation, we propose the Fixed-points of Restoration priors (FiRe) framework, which enables using any restoration model coupled with the degradation it was trained on as an implicit prior to solve \eqref{eq:inv_pb}.

Our contributions are as follows:
first, we demonstrate that incorporating restoration models trained on general restoration tasks requires 
to combine them with an appropriate degradation step, and show it can lead to well behaved iterative algorithms.
Second, we derive a closed-form expression for the prior associated with general restoration models, providing a theoretical connection with the PnP framework.
Third, we show that our formulation naturally accommodates additive priors, enabling both ensemble-based approaches and acquisition-aware conditioning of the restoration process.
Finally, our experiments demonstrate that the proposed framework significantly improves reconstruction quality across various inverse problems.

\section{Background and related works}

\subsection{Denoising priors and PnP algorithms}

While early methods proposed explicit priors derived from signal-processing arguments (\emph{e.g.}\ sparsity in a wavelet basis \cite{stephane1999wavelet} or of the gradient image \cite{bredies2010total}), recent progress in imaging inverse problems relies on implicit denoising-based priors and their connection to score-matching \cite{kadkhodaie2021stochastic, daras2024survey}.
Minimum Mean Square Error (MMSE) denoisers play a crucial role to implicitly probe the prior distributions of smooth images. For a given noisy sample $y= x+n$ where $x\sim p(\text{data})$ and $n\sim \mathcal{N}(0, \sigma^2)$, these operators $\operatorname{R}$ satisfy $\operatorname{R}(y) = \mathbb{E}[x|y]$.
Starting from this conditional expectation, Tweedie's formula yields \cite{reehorst2018regularization, xu2020provable, laumont2022bayesian}
\begin{equation}
\label{eq:tweedie}
\nabla_y \log p_\sigma(y) = (\operatorname{R}(y)-y)/\sigma^2,
\end{equation}
thus providing a link between the denoiser $\operatorname{R}$ and the sought prior.
In turn, defining $f(x) = \frac{1}{2}\|Ax-y\|^2$, and assuming that $-\nabla \log p$ is well approximated by a neural network $\operatorname{R}$, the popular RED approach suggests the following iteration
\begin{equation}
\label{eq:basic_red}
    x_{k+1} = \operatorname{prox}_{\lambda f}(x_k-\gamma \operatorname{R}(x_k)),
\end{equation}
where $\gamma,\lambda>0$ are stepsizes, and where the $\operatorname{prox}_{\lambda f}$ denotes the proximity operator\footnote{Defined as $\operatorname{prox}_{\lambda f}(u) =\underset{x}{\operatorname{argmin}}\,\, \lambda f(x) + \frac{1}{2}\|x-u\|_2^2$.} of $\lambda f$.

Alternative approaches exist, leveraging slightly different assumptions.
For instance, Plug-and-Play (PnP) algorithms rather assume that the denoiser $\operatorname{R}$ is the proximal operator of $-\log p$ \cite{pesquet2021learning, hurault2022proximal},
resulting in algorithms where the role of $f$ and $\operatorname{R}$ are swapped in \eqref{eq:basic_red}.
In the same spirit, the modified Half-Quadratic Splitting (HQS) scheme
\begin{equation}
\label{eq:hqs_basic}
    x_{k+1} = \operatorname{R}(\operatorname{prox}_{\lambda f}(x_k)),
\end{equation}
was shown to achieve excellent performance \cite{zhang2021plug}.

While these approaches have shown promising results, the methods in \eqref{eq:basic_red} and \eqref{eq:hqs_basic} often yield suboptimal reconstructions. 
We argue that this limitation stems from two key issues: 
first, during inference, the denoiser processes inputs that deviate significantly from its training distribution of Gaussian noise-corrupted images. This mismatch is particularly problematic at convergence, where the denoiser's action must align with the score of the inverse problem's noise to yield stable solutions. Second, and more fundamentally, smooth images are not fixed-points of denoising networks: even for simple denoising problems where $A = \operatorname{Id}$, iterating $\operatorname{R}$ in \eqref{eq:hqs_basic} leads to artifacts, similar to those shown in \autoref{fig:summary} \textcolor{cvprblue}{(b)}. In this work, we adopt the view-point of the fixed-point sequence generated by \eqref{eq:basic_red} to overcome the shortcomings of this approximate Tweedie formula as well as the distribution shift at inference time.

\subsection{Related works}

\paragraph{Implicit denoising priors}
The link between denoising autoencoders and score matching was first established in \cite{vincent2011connection}, showing that denoising autoencoders implicitly estimate the score $\nabla \log p(x)$  of the data distribution $p(x)$. Implicit denoising priors were introduced by \cite{venkatakrishnan2013plug}, who proposed replacing classical proximity operators with denoising operators in ADMM algorithms, interpreting the denoising operation as a MAP estimator. This framework was later revisited by \cite{reehorst2018regularization, xu2020provable}, who formalized the connection through Tweedie's formula \eqref{eq:tweedie} for MMSE denoisers and provided a proximal operator interpretation. Subsequent work by \cite{laumont2022bayesian} deepened the theoretical understanding of the underlying distributions and proposed a PnP version of Langevin sampling. Alternative viewpoints emerged through the study of monotone operators \cite{sun2019onlinepnp, pesquet2021learning}, and modifications of the denoising neural network architectures for nonconvex priors \cite{hurault2022proximal}. Our theoretical analysis diverges from previous approaches that rely on Lipschitz properties \cite{pesquet2021learning, hurault2022proximal, xu2020provable}, instead building upon fixed-point arguments \cite{cohen2021regularization}.

\paragraph{General restoration priors}
The extension of implicit priors beyond denoising remains relatively unexplored. The early work \cite{zhang2019deep} considered super-resolution models as priors for image super-resolution. In the slightly different context of generative modeling trough diffusion, the authors of \cite{bansal2024cold} proposed to replace the denoising diffusion process with a general restoration diffusion process.  Yet, these works did not establish a clear link with the variational approach and namely no explicit prior was derived.
Recently, significant progress was made in DRP \cite{hu2023restoration}: after replacing denoisers with super-resolution models, the authors were able to derive a Tweedie-like formula under specific degradation assumptions.
The subsequent ShaRP algorithms \cite{hu2024stochastic} waived specific restrictions but leveraged the same prior formula.
While these works made it possible to use a broader range of restoration networks, they are not drop-in replacements of denoisers in PnP and can fail when they are not suited to a given task.

\paragraph{Expected and ensembled priors} Recent works have explored the benefits of considering priors in expectation. In \cite{renaud2024plug}, the authors introduced the concept of expected implicit priors, where at the algorithmic level, a noisy sample of the score is used at each step.
Recently, ShaRP \cite{hu2024stochastic} extended DRP to handle random degradation models. In a different direction, \cite{terris2023equivariant} leveraged equivariant properties of the denoising neural network, leading to averaged invariant priors. A common point of these algorithms is, however, that random operations are introduced at each step of the algorithm to slightly perturbate the behaviour of the denoiser.
While ensembling techniques are relatively unexplored in image restoration, with few exceptions such as \cite{jiang2019ensemble}, the recent work by \cite{sun2024ensir} has investigated post-training ensembling strategies for end-to-end neural network reconstructions.

\section{Proposed framework}

\subsection{Fixed points of restoration models}

A central observation to our work is that image restoration networks inherently possess a fixed-point property with respect to smooth images. More precisely,
consider a general restoration problem $y = Hx + w$, where $H$ is a linear operator and $w$ denotes additive noise. In the standard supervised setting, image restoration models are trained by minimizing an objective of the form
\begin{equation}
\label{eq:sup_training}
\mathcal{L}(\theta) = \mathbb{E}_{x\sim p_\text{data}, w\sim \mathcal{W}}\left[\|\operatorname{R}_\theta(Hx + w)-x\|\right],
\end{equation}
where $p_\text{data}$ denotes the distribution of natural images, $\mathcal{W}$ is the noise distribution, and $\|\cdot\|$ typically denotes the $\ell_2$ or $\ell_1$ norm.
This loss leads to an important property: let $C = \{ x\in\mathbb{R}^n ; T(x) = x\}$ denote the set of fixed-points of $T = \operatorname{R} \circ D$ with $D = (H \cdot + w)$ the degradation operator. 
By minimizing the training objective \eqref{eq:sup_training}, 
one implicitly ensures that restored images lie close to the set $C$.
Moreover, images restored with \eqref{eq:sup_training} tend to be smooth \cite{blau2018perception, ohayon2024perception}.
This observation suggests an approximate idempotence property $T^2\approx T$, motivating the following assumption:

\begin{assumption}
\label{ass:proj}
The operator $T = \operatorname{R} \circ D$ can be expressed as a projection $T=\operatorname{proj}_C$ onto a closed, prox-regular set $C$.
\end{assumption}

\noindent In this assumption, we have used the following definition: a closed set $C$ is prox-regular if each point in $C$ has a neighborhood where the projection onto $C$ is single-valued \cite{lewis2009local}. Although the prox-regularity of $C$ might be too restrictive for smooth images, this assumption is weaker than convexity and it is satisfied by a broad class of sets including smooth submanifolds of $\mathbb{R}^n$ \cite{poliquin2010calculus}. The following result, from \cite{lewis2009local}, allows to link the projection on the set $C$ to the gradient of a smooth function.

\begin{proposition}
Under \autoref{ass:proj}, around any point $x$ at which the set $C$ is prox-regular, the operator $T$ satisfies $T(x) = x-\frac{1}{2}\nabla d_C^2(x)$, where $d_C$ denotes the distance function to the set $C$, i.e. $d_C(x) = \underset{u\in C}{\text{inf}}\,\,\|x-u\|$. Furthermore, $\nabla d_C^2$ is 2-Lipschitz.
\label{prop:grad}
\end{proposition}

\noindent As a consequence, under \autoref{ass:proj}, the restoration model composed with the associated degradation model naturally writes as the gradient of an explicit functional, namely the distance to the fixed-points of $T = \operatorname{R}\circ D$.

Interestingly, this analysis suggests that restoration models only exhibit meaningful fixed-points behaviour when composed with the degradation operator for which they have been trained.
This phenomenon is illustrated in \autoref{fig:summary}.
It shows the evolution of the PSNR and the end points of the sequences $x_{k+1} = \operatorname{R}(x_k)$ and $x_{k+1} = \operatorname{R}(Hx_k + w_k)$, for two pairs of restoration model / degradation model.
On the top row, we choose $\operatorname{R}$ as the SCUNet trained on multiple degradations, and compose it with the (non-linear) JPEG degradation operator $H$.
In the bottom row, we choose a Restormer model trained on Gaussian deblurring and choose $H$ as a Gaussian blur operator.
We observe that the sequence $x_{k+1} = \operatorname{R}(x_k)$ quickly yields images with strong artifacts and low PSNR, while the sequence motivated by our analysis $x_{k+1} = \operatorname{R}(Hx_k+w_k)$ yields stable behaviours in both cases, with convergence towards an artifacts-free natural image.

In practice, restoration models only approximately minimize \eqref{eq:sup_training} over a set of natural images. Therefore, the fixed point set associated to one restoration model may contain slightly degraded images, or smooth images might only be in its neighborhood.
To overcome this difficulty, we propose to rely on an ensemble of restoration priors.
More precisely, for a given restoration model $\operatorname{R}_\xi$ and degradation $D_\xi$, we denote $C_\xi$ the associated fixed-point set.
While $C_\xi$ may contain spurious fixed-points due to the approximate nature of the training, these spurious elements are unlikely to be shared with sets $C_{\xi'}$ associated with models $\operatorname{R}_{\xi'}$ trained for different tasks $D_{\xi'}$.
On the contrary, we expect that smooth images should lie in the neighborhood of all these sets.
By considering the intersection $\bigcap_\xi C_\xi$, we can obtain a better approximation of the set $C$ containing smooth images, as this intersection filters out model-specific artifacts.
Following \autoref{prop:grad}, this naturally leads to considering the log-prior $\mathbb{E}_{\xi}[d_{C_\xi}^2(x)]$, which measures how far an image $x$ is from the intersection of the sets $C_\xi$.
With this log-prior, even if the sets' intersection is empty, the maximum likelihood points lie the closest possible to each of these sets.

\subsection{Proposed FiRe framework}
\label{sect:proposed}

In order to find inverse problem solutions close to the fixed point sets $C_\xi$ of the degradation-restoration models $\operatorname{R}_\xi\circ D_\xi$,
we propose the Fixed-points of Restoration priors (FiRe) framework. In essence, the FiRe framework leverages the distance to these fixed-point sets in \eqref{eq:bayes} as:
\begin{equation}
\label{eq:general_min_exp}
x^* = \underset{x}{\text{argmin}}\,\, \lambda f(x) + \frac{\gamma}{2}\mathbb{E}_{\xi \sim \Xi}\left[d_{C_\xi}^2(x)\right],
\end{equation}
where $\lambda>0$ and $0 < \gamma < 1$ are hyperparameters and $\Xi$ represents the sampling distribution over different pairs of degradation-restoration models.
Therefore, our framework uses the variational formulation \eqref{eq:bayes} with the associated prior
\begin{equation}
p(x) \propto \exp\left(-\small{\frac{1}{2}}\mathbb{E}_{\xi \sim \Xi}\left[d_{C_\xi}^2(x)\right]\right),
\end{equation}
establishing a connection with traditional PnP formulations.

To minimize \eqref{eq:general_min_exp}, \autoref{prop:grad} provides an explicit expression for the gradient of each restoration prior.
Indeed, the set of fixed points $C_\xi$ is implicitly defined for each pair of restoration model $\operatorname{R}_\xi$ and its associated degradation operator $D_\xi$, and we have $\frac12\nabla d_{C_\xi}^2(x) = x - \operatorname{R}_\xi(D_\xi(x))$.
Based on this expression, we propose the FiRe-HQS algorithm, described in \autoref{alg:sto_sgd}, for solving \eqref{eq:general_min_exp}.
At its core, FiRe extends the PnP framework by incorporating these degradation steps before applying each restoration model and enforcing data consistency through a proximal step.
By sampling at each iteration a degradation-restoration model, we obtain cheap stochastic updates toward the overall objective.
This algorithm also supports averaging multiple parallel updates originating from varied pairs of restoration models and their associated degradation operators, improving the prior quality.
The core difference between this algorithm and traditional PnP methods is that the composition of restoration and degradation operators ensures that restoration models' inputs match the distribution used to train the model.
This justifies and generalizes the recently proposed noise injection strategy~\cite{renaud2024plug} for denoisers.

\begin{algorithm}[t]
\caption{FiRe-HQS algorithm}
\label{alg:sto_sgd}
\begin{algorithmic}[1]
\State \textbf{Input:} $x_0$, $\gamma_n$, $\lambda$, restoration models $(\operatorname{R}^1, \hdots, \operatorname{R}^N)$ associated to degradations $(\mathcal{D}^1\!,  \hdots, \mathcal{D}^N)$.
\For{$k=1, \hdots, K$}
    \For{$n=1, \hdots, N$}
    \State Select restoration model $\operatorname{R}^n$
    \State $\text{Sample } (H_k^n, w_k^n) \sim \mathcal{D}^n$
    \State $r_k^n = x_k-\operatorname{R}^n(H_k^nx_k+w_k^n)$
    \EndFor
    \State $u_{k} = x_k- \sum_{n=1}^N\gamma_n r_k^n$ 
    \State $x_{k+1} = \operatorname{prox}_{\lambda f}(u_k)$ 
\EndFor
\State \textbf{Return:} $x_{k+1}$
\end{algorithmic}
\end{algorithm}
\vspace{-1em}

\begin{table*}[t]
	\centering
    \resizebox{\textwidth}{!}{
	\begin{tabular}{lllcccccccccccccccccc}
		\toprule
		\multirow{2}{*}{{Dataset}} & \multirow{2}{*}{{Method}} & \multirow{2}{*}{{Restoration prior}} & \multicolumn{3}{c}{Gaussian Blur} & \multicolumn{3}{c}{Motion Blur} & \multicolumn{3}{c}{SRx4} & \multirow{2}{*}{{Iters}} \\ 
        \cmidrule(lr){4-6} \cmidrule(lr){7-9} \cmidrule(lr){10-12}
        & & & PSNR$\uparrow$&\!\!\!SSIM$\uparrow$\!\!\!&LPIPS$\downarrow$ & PSNR$\uparrow$&\!\!\!SSIM$\uparrow$\!\!\!&LPIPS$\downarrow$ & PSNR$\uparrow$&\!\!\!SSIM$\uparrow$\!\!\!&LPIPS$\downarrow$ & \\
        \specialrule{\lightrulewidth}{1pt}{0pt}
        \rowcolor{blue!5}
        & DRP & SRx2 + SRx3 & \underline{25.48} & \underline{0.76} & \underline{0.30} &  25.38 & 0.75 & 0.30 &  23.23 & 0.67 & 0.44 & 30 \\
        \rowcolor{blue!5}
        & \textbf{FiRe-HQS} & denoising + SR + deblur &  \textbf{25.80} & \textbf{0.77} & \textbf{0.30} & \textbf{30.49} & \textbf{0.91} & \textbf{0.10} & \textbf{23.92} & \textbf{0.69} & \textbf{0.39} & 30\\[1pt]
        \cline{2-13}
        \rowcolor{blue!5}
        & DPIR & denoising & 25.18 & 0.74 & 0.40 & \underline{30.39} & \underline{0.88} & 0.17 & 23.60 & \underline{0.68} & 0.46 & 20 \\
        \rowcolor{blue!5}
        \multirow{-4}{*}{Imnet100} & \multirow{1}{*}{DiffPIR} & denoising & 25.32 & 0.73 & 0.36 & 29.70 & 0.86 & \underline{0.12} & \underline{23.89} & 0.66 & \underline{0.41} & 20  \\[1pt]
        \specialrule{\lightrulewidth}{0pt}{0pt}
        \rowcolor{orange!10}
        & \multirow{1}{*}{DRP} & SRx2 + SRx3 & 26.18 & \underline{0.79} & \underline{0.29} & 26.05 & 0.78 & 0.28 & 23.65 & 0.69 & 0.44 & 30 \\
        \rowcolor{orange!10}
        & \multirow{1}{*}{\textbf{FiRe-HQS}} & denoising + SR + deblur & \textbf{27.00} & \textbf{0.80} & \textbf{0.27} & \textbf{31.67} & \textbf{0.91} & \textbf{0.10} & \underline{25.15} & \underline{0.71} & \underline{0.39} & 30 \\[1pt]
        \cline{2-13}
        \rowcolor{orange!10}
        & \multirow{1}{*}{DPIR} & denoising & 26.14 & 0.77 & 0.40 & \underline{30.95} & 0.87 & 0.20 & 24.64 & 0.70 & 0.46 & 20 \\
        \rowcolor{orange!10}
        \multirow{-4}{*}{BSD20}  & \multirow{1}{*}{DiffPIR} & denoising & \underline{26.67} & 0.76 & 0.36 & 30.85 & \underline{0.88} & \underline{0.13} & \textbf{26.18} & \textbf{0.79} & \textbf{0.29} & 20 \\[1pt]
        \specialrule{\lightrulewidth}{0pt}{0pt}
	\end{tabular}}
	\caption{Image restoration results on different problems with various algorithms on (\emph{top}) Imnet100 test set and (\emph{bottom}) BSD20 subset. Here, the proposed FiRe-HQS algorithm is implemented with a combination of denoising (SCUNet), SR$\times$2 (SwinIR$\times$2), and Gaussian deblurring (Restormer) restoration models. We compare it to another restoration prior (DRP) and two PnP algorithms (DPIR and DiffPIR).}
	\label{tab:results_combined}
\end{table*}

\paragraph{Difference with DRP \& ShaRP}
The use non-denoising restoration models as priors for inverse problems was first proposed in the DRP algorithm \cite{hu2023restoration} and its recent extension ShaRP \cite{hu2024stochastic}.
Both derive from the following prior
\begin{equation}
p(x) \propto \mathbb{E}_{s \sim G_\sigma(s-Hx), H \sim p_H} \left[ -\log p(s|H)\right],
\end{equation}
where $p_H$ is the distribution of degradation operators and $G_\sigma$ is the Gaussian density function.
The key difference lies in the corresponding score function
\begin{equation}
\label{eq:sharp_gradient}
-\nabla_x\log p(x) \propto H^\top H(x-\operatorname{R}(Hx + w)),
\end{equation}
where $\operatorname{R}$ is the restoration operator. While ShaRP directly uses this gradient, DRP exhibits a similar $H^\top H$ term which is critical for its scaled-proximal operator interpretation. Yet, this confines the information obtained through the restoration-based prior to $\ker(H)^\perp$. This limitation is restrictive since it is precisely within $\ker(H)$ that the model $\operatorname{R}$ expresses the learned implicit prior, as demonstrated in our experiments (see \emph{e.g.} \autoref{fig:inpainting_comparison}).

\section{Properties of the FiRe framework}

We next provide theoretical properties for the proposed algorithm, in both deterministic and stochastic settings.

\subsection{Deterministic restoration priors}

In this section, we consider a finite-sum version of \eqref{eq:general_min_exp}, which occurs for instance in the case of a finite number of restoration models $(\operatorname{R}^n)_{1\leq n\leq N}$ associated to a finite number of deterministic degradation models $(H^n)_{1\leq n \leq N}$. In this context, \eqref{eq:general_min_exp} rewrites as
\begin{equation}
\label{eq:general_finitesum_exp}
x^* = \underset{x}{\text{argmin}}\,\, \lambda f(x)
    + \frac1{2}\sum_{n=1}^N \gamma_nd_{C_n}^2(x),
\end{equation}
where $C_n$ is the fixed-point set associated to the pair $(\operatorname{R}^n, H^n)$.
Furthermore, at each iteration $k$, the degradation operator is fixed, \emph{i.e.} $H_k^n = H^n$, and we set $w_k^n = 0$.

\begin{proposition}
Assume that, for all $n\in\{1,\hdots, N\}$, $\operatorname{R}^n$ satisfies \autoref{ass:proj} holds for some $C_n$ that is uniformly prox-regular.
Assume furthermore that, for all $n$, $\gamma_n \in [0, 1]$ and $\sum_{n=1}^N\gamma_n < 1$.
Then the sequence $(x_k)_{k\in\mathbb{N}}$ generated by \autoref{alg:sto_sgd} converges to a point $x^*$ satisfying \eqref{eq:general_finitesum_exp}.
\end{proposition}

\begin{proof}
Without loss of generality, we assume that $N=1$. \autoref{prop:grad} allows us to rewrite line 6 of \autoref{alg:sto_sgd} as 
\begin{equation}
\begin{aligned}
    u_{k} &= \gamma \operatorname{proj}_C(x_k) + (1-\gamma)x_k \\
    &= \gamma (x_k-\frac{1}{2}\nabla d_C^2(x_k))+(1-\gamma)x_k \\
    &= x_k - \frac{\gamma}{2} \nabla d_C^2(x_k).
\end{aligned}
\end{equation}
Therefore, \autoref{alg:sto_sgd} rewrites $x_{k+1} = \operatorname{prox}_{\lambda f}(x_k-\frac{\gamma}{2} \nabla d_C^2(x_k))$. $f$ is convex and \cite{lewis2009local} yields that $d_C^2(x_k)$ has Lipschitz gradient. Therefore, with $\gamma < 1$, we are in the setting of \cite[Theorem 5.1]{attouch2013convergence} and the result follows.
\end{proof}

We stress that, when the intersection $\bigcap_n C_n \neq \emptyset$ and additional regularity condition on the intersection hold, a version of Algorithm \ref{alg:sto_sgd} with $f=0$ would yield a point in the intersection of all sets \cite[Theorem 3.5]{attouch2013convergence}, \emph{i.e.} in our case, the fixed points of all restoration networks. In other words, when the restoration networks share at least one common fixed point, our method can naturally discover such a point by iteratively applying each network's transformation.

\subsection{Expected prior}
\label{sect:expected_prior}

In practice, except for particular restoration networks, the models $\operatorname{R}_\theta$ are often trained to solve a class of problems. For instance, for various blurs generated at random, or different noise levels.
In this case, the supervised training loss \eqref{eq:sup_training} takes the form
\begin{equation}
    \widetilde{\mathcal{L}}(\theta) = \mathbb{E}_{x\sim p_\text{data}}
    \mathbb{E}_{(H,w)\sim\mathcal{D}}\|\operatorname{R}_\theta (Hx+w)-x\|,
\end{equation}
with the measurement $H$ and the noise $w$ generated at random from a class of degradation $\mathcal{D}$ \cite{zhang2021plug, zhang2023practical}.
Therefore, the set of fixed points $C$ is defined in expectation over the degradations $(H, w) \sim \mathcal D$.

Results from \autoref{prop:grad} can be extended to this averaged case as follows: we assume that there exists a closed, prox-regular set $C$ such that
\begin{equation}
\label{eq:exp_grad}
\begin{aligned}
\frac{1}{2}\nabla d_C^2(x) = x - \mathbb{E}_{(H,w)\sim \mathcal D}[\operatorname{R}(Hx + w)].
\end{aligned}
\end{equation}
Using random realisations of $H_k$ and $w_k$ at each iteration $k$, we can thus build stochastic estimate $g_k$ of \eqref{eq:exp_grad} as
\begin{equation}
\label{eq:grad_sto}
    g_k = x_k - \operatorname{R}(H_kx_k+w_k).
\end{equation}
In this case \autoref{alg:sto_sgd} reduces to an instance of the stochastic proximal gradient descent algorithm. The following result, which is made more precise in the Appendix, can be derived from \cite{li2022unified}.

\begin{proposition}
Assume that for all $n$ and $k$, $g_k^n$ is a unbiased estimate of $\mathbb{E}_\xi[\nabla d_{C_\xi}^2(x,\xi)]$ with bounded variance; assume furthermore that the stepsizes $(\gamma_k^n)_{k\in\mathbb{N}}$ are decreasing and non-summable.
Defining the residual function $F(x) = x-\operatorname{prox}_{\lambda f}(x-\frac{1}{2}d_C^2(x))$, we have that $\mathbb{E}[F(x_k)]\underset{k\to\infty}{\longrightarrow} 0$.
\end{proposition}

\section{Experimental results}

We now investigate the proposed algorithm for solving different instances of \eqref{eq:inv_pb}, namely Gaussian deblurring, motion deblurring, single image super-resolution and image inpainting, with several backbone restoration priors.

\subsection{Considered priors \& baselines}

To show the generality of the proposed FiRe framework, we instanciate it with various degradation-restoration models readily available from the \texttt{deepinv} library\footnote{\url{https://deepinv.github.io/}} \cite{tachella2023deepinv} and the \texttt{lama-project}\footnote{\url{https://advimman.github.io/lama-project/}}.
\begin{itemize}
\item DRUNet \cite{zhang2021plug}: we consider the standard DRUNet model trained for Gaussian denoising. In this case, $H = \operatorname{Id}$ and $w \sim \mathcal{N}(0, \sigma^2)$.
\item Restormer \cite{zamir2022restormer}: we consider 2 versions of the Restormer that were finetuned on Gaussian and Motion deblurring respectively. In this case, $H$ is the set of Gaussian (resp. motion) blurs, and $w \sim \mathcal{N}(0, \sigma^2)$.
\item SCUNet \cite{zhang2023practical}: we consider the SCUNet trained for general non-linear image restoration problems. We study $H = \operatorname{JPEG}_{q}$ with quality factors $q\in[20, 100]$, and $w \sim \mathcal{N}(0, \sigma^2)$.
\item SwinIR \cite{liang2021swinir}: we evaluate two variants trained for super-resolution with scale factors $\times2$ and $\times3$ respectively. In this case, $H$ consists of the corresponding downsampling operators, and $w=0$.
\item LAMA \cite{suvorov2022resolution}: we consider 2 versions of the LAMA model, the one provided by \cite{suvorov2022resolution} as well as one fine-tuned on random inpainting. In this case, $H$ is a large mask operator (resp. random missing pixels), and $w=0$.
\end{itemize}
\noindent Details regarding the finetunings can be found in Appendix.

\noindent We evaluate our method against several baselines. First, we compare with DRP \cite{hu2023restoration}, which is most similar to our approach but relies on a specific super-resolution prior (SwinIR). Second, we include a standard PnP baseline using the DRUNet denoising model. Finally, we compare against state-of-the-art methods: DPIR \cite{zhang2021plug}, representing classical PnP approaches, and DiffPIR \cite{zhu2023denoising}, providing a fair comparison with diffusion-based methods.

\begin{table}[t]
	\centering
	\small
    \resizebox{0.47\textwidth}{!}{
	\begin{tabular}{@{\hskip 0pt}c@{\hskip 8pt}l@{\hskip 6pt} l@{\hskip 6pt} c@{\hskip 6pt} c@{\hskip 6pt} c@{\hskip 0pt}} 
		\toprule 
        & Backbone & Restoration prior & \makecell{Gaussian\\Blur} & \makecell{Motion\\Blur} & SR$\times$4 \\ 
		\midrule
        & \multirow{1}{*}{none} & - & 21.57 & 19.48 & 21.69 \\
  \hline
\multirow{2}{*}{\rotatebox[origin=c]{90}{PnP}} & \multirow{1}{*}{DRUNet} & denoising & 25.93 & 27.95 & 22.54 \\
        & \multirow{1}{*}{SCUNet} & blind denoising & 23.91 & 30.19 & 23.05 \\
        \hline
        \multirow{3}{*}{\rotatebox[origin=c]{90}{\textcolor{black}{ShaRP}}} & \multirow{1}{*}{\textcolor{black}{LAMA}} & \textcolor{black}{random inpainting} & \textcolor{black}{23.42} & \textcolor{black}{23.27} & \textcolor{black}{21.96} \\
& \multirow{1}{*}{\textcolor{black}{Restormer}} & \textcolor{black}{gaussian debluring} & \textcolor{black}{27.14} & \textcolor{black}{28.12} & \textcolor{black}{23.16} \\
& \multirow{1}{*}{\textcolor{black}{SwinIR}} & \textcolor{black}{SR$\times$2} & \textcolor{black}{26.01} & \textcolor{black}{26.66} & \textcolor{black}{23.24} \\
\hline
\multirow{8}{*}{\rotatebox[origin=c]{90}{FiRe}} & \multirow{1}{*}{LAMA} & random inpainting & 24.72 & 26.19 & 21.80 \\
       & & random brush inpainting & 24.37 & 24.35 & 21.83 \\
       & \multirow{1}{*}{Restormer} & gaussian deblurring & 27.01 & 29.83 & 24.29 \\
       & & motion deblurring & 25.92 & 29.17 & 21.17 \\
       & \multirow{1}{*}{SCUNet} & blind denoising & \underline{28.03} & \underline{30.84} & \underline{24.30} \\
       & & JPEG restoration & 27.19 & \textbf{31.23} & 24.01 \\
      & \multirow{1}{*}{SwinIR} & SR$\times$2 & \textbf{28.19} & 27.17 & 24.02 \\
      & & SR$\times$3 & 27.31 & 26.84 & \textbf{24.77} \\
		\bottomrule
	\end{tabular}}
    \caption{PSNR values on the Set3C dataset for different restoration problems. The first row (``none'') corresponds to the pseudo-inverse reconstruction baseline. The second row shows results from a standard PnP approach using DRUNet. Third row shows results with the ShaRP algorithm. Subsequent rows present results from \autoref{alg:sto_sgd} using different restoration models (LAMA, Restormer, SCUNet, SwinIR) and their associated training tasks. Best results per column in \textbf{bold}, second best \underline{underlined}.}
	\label{tab:results_mini}
\vspace{-1em}
\end{table}

\subsection{Single restoration prior}

We first evaluate \autoref{alg:sto_sgd} in its simplest form using a single restoration prior ($N=1$). \autoref{tab:results_mini} reports PSNR values on the Set3C dataset for various restoration models and their associated training tasks. For comparison, we include two PnP baselines: the DRUNet denoising model and the blind SCUNet model. The PnP framework is recovered within Algorithm~\ref{alg:sto_sgd} by setting $H_k=\operatorname{Id}$ and $w_k = 0$.
Most restoration priors within our FiRe framework achieve comparable or better performance than the DRUNet baseline, with the exception of LAMA inpainting priors (which still outperform the pseudo-inverse reconstruction).
Also note that the SCUNet blind denoiser performs significantly better within our FiRe framework than in the PnP setting.
This suggests that our fixed-point formulation successfully generalizes the denoising prior paradigm to a wide class of restoration models.
We also compare to ShaRP using the same restoration priors. To achieve good performance with ShaRP required increasing the iteration limit to 200 (vs. 30 for FiRe), significantly raising the computational cost.
\autoref{fig:sr4_bsd_scunet} shows the restoration result for the SCUNet JPEG prior in a PnP vs FiRe framework. Note how the latter reconstruction has less artifacts. 

Associated visual for the SR$\times4$ problem are shown in \autoref{fig:sr4_vanilla}. We notice that except the LAMA inpainting priors, all priors perform on par with DPIR, and while all priors seem to yield similar reconstructions, few stand out. The LAMA prior seems to introduce glitch-like artifacts in the reconstruction, while the motion deblurring prior seems to introduce motion-like artifacts. All other priors seem to slightly over-smooth the reconstruction. The fact that the LAMA inpainting prior enables meaningful restoration is surprising, since it was only trained on binary mask inpainting. In particular, setting all mask values to one would reduce the problem to simple least squares minimization (\emph{i.e.}, pseudo-inverse restoration), whereas introducing sparse masks enables the model to leverage its learned prior. The model's inability to remove noise, which we attribute to its noiseless training regime, further demonstrates how the training setup influences the prior's behavior. This observation supports our broader claim that restoration models can impose meaningful priors even when trained on tasks significantly different from the target problem.

We provide in~\cref{fig:ablation_study} results showing the influence of the degradation strength on the reconstruction quality. We observe that the strength of the degradation plays a role similar to that of the regularization parameter in traditional variational methods.

\begin{figure}[t]
\small
    \centering
    \setlength{\tabcolsep}{1pt}
    \begin{tabular}{@{}c@{\hskip 4pt}c@{\hskip 4pt}c@{}}
    $y$ & PnP & FiRe \\
        \includegraphics[width=0.15\textwidth]{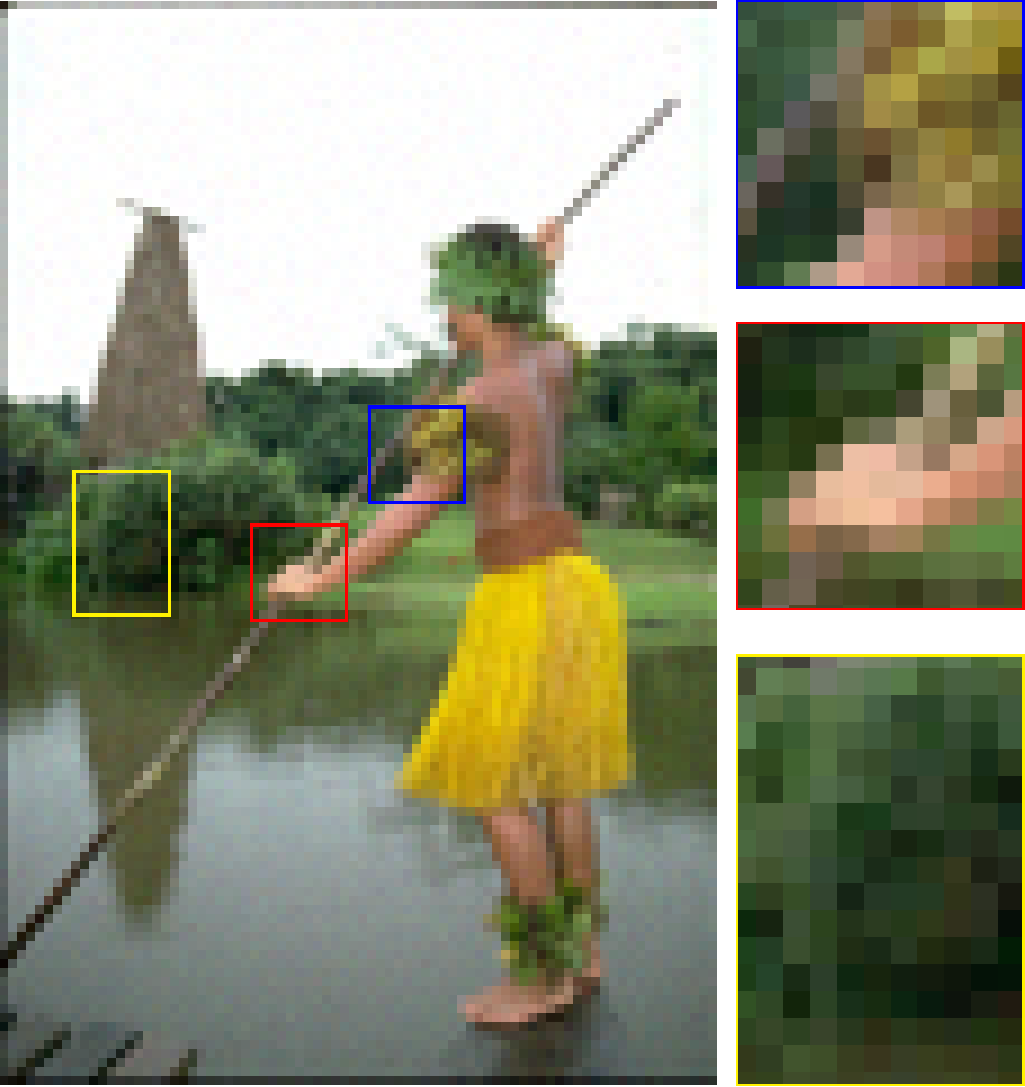} &
        \includegraphics[width=0.15\textwidth]{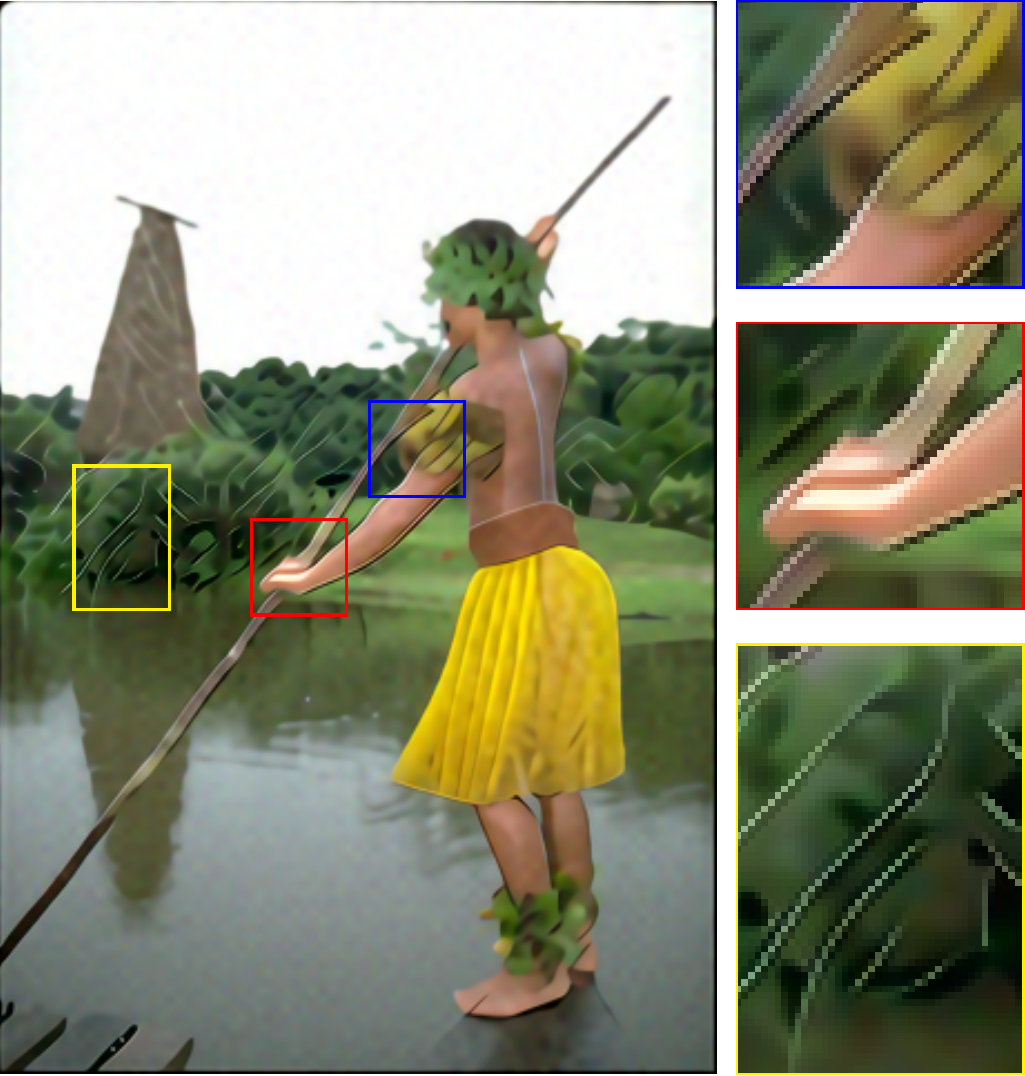} &
        \includegraphics[width=0.15\textwidth]{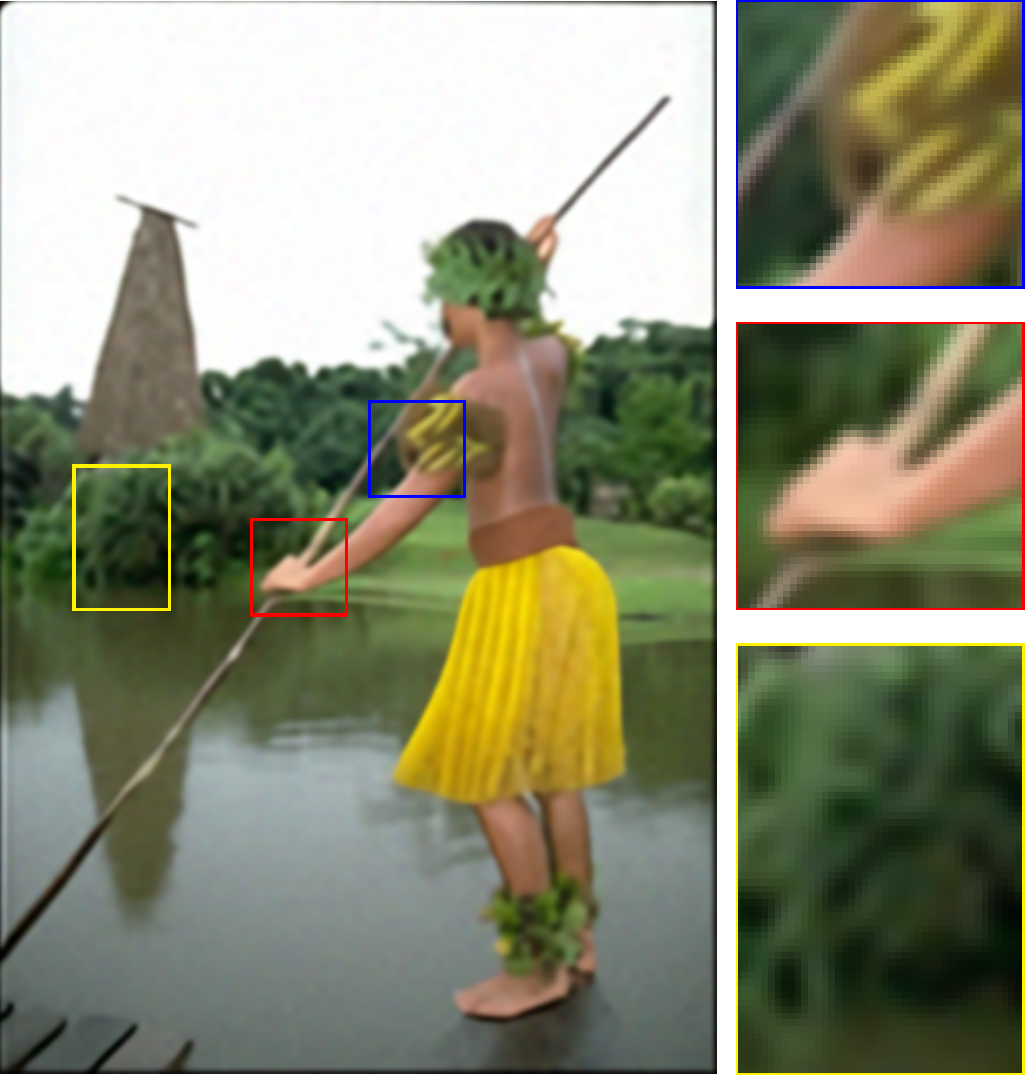} \\
        PSNR & 25.45 dB & 25.96 dB 
    \end{tabular}
    \vspace{-1em}
    \caption{Reconstruction results on the SR$\times$4 problem with the SCUNet JPEG prior.}
    \label{fig:sr4_bsd_scunet}
\vspace{-1em}
\end{figure}

\begin{figure}[t]
\small
    \centering
    \setlength{\tabcolsep}{1pt}
    \begin{tabular}{@{}cccc@{}}
        (a) $y$ & (b) $A^\dagger y$ & (c) DPIR \cite{zhang2021plug} & (d) LAMA \\
        \includegraphics[width=0.12\textwidth]{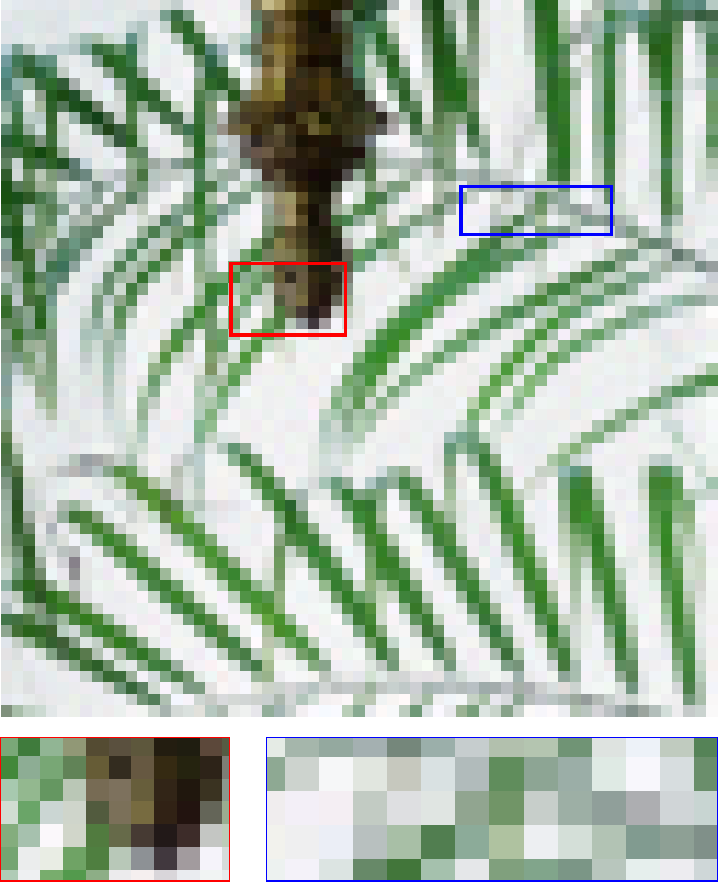} &
        \includegraphics[width=0.12\textwidth]{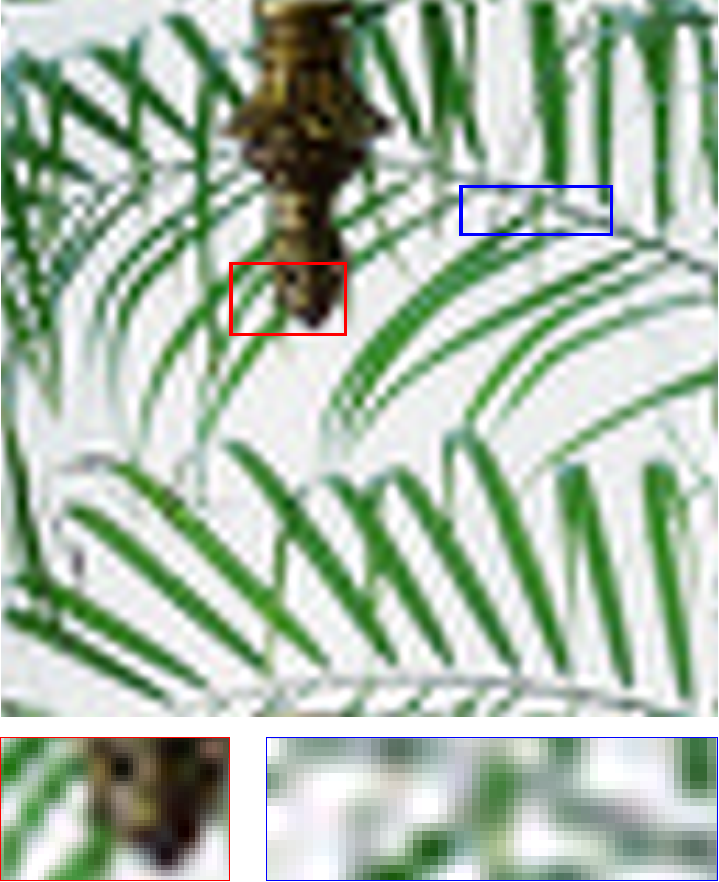} &
        \includegraphics[width=0.12\textwidth]{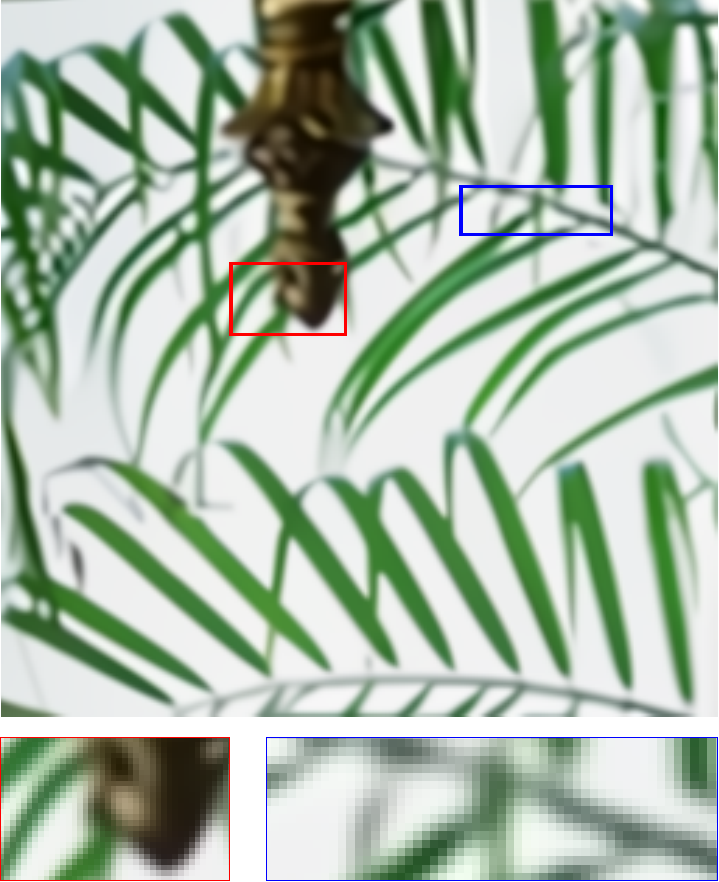} &
        \includegraphics[width=0.12\textwidth]{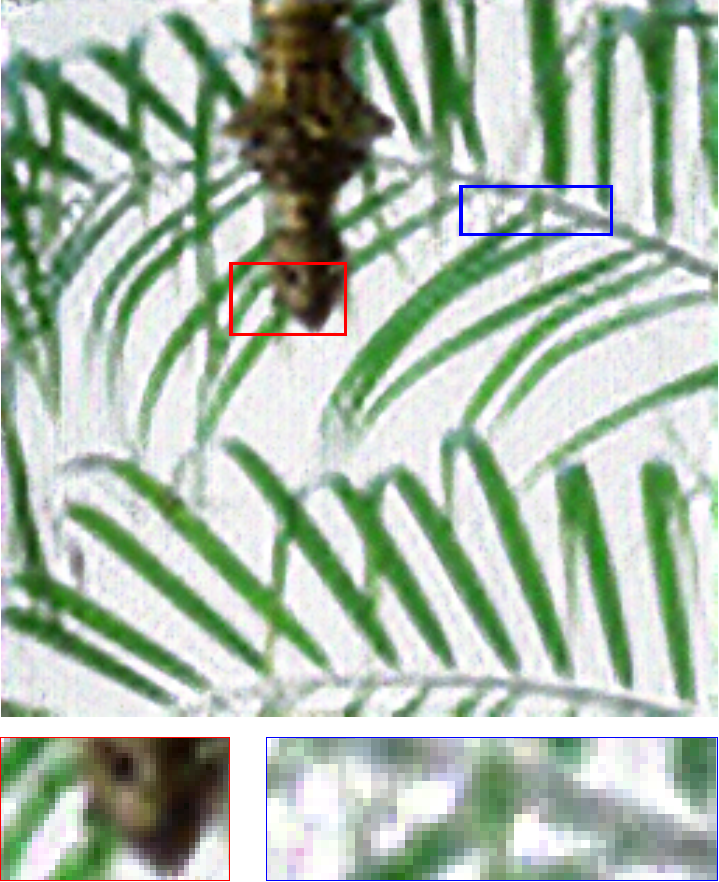} \\
         PSNR & 19.71 dB & 22.04 dB & 20.27 dB \\[6pt]
        (e) Rest. Motion & (f) SwinIR 2$\times$ & (g) SCUNet JPEG & (h) Rest. gauss. \\
        \includegraphics[width=0.12\textwidth]{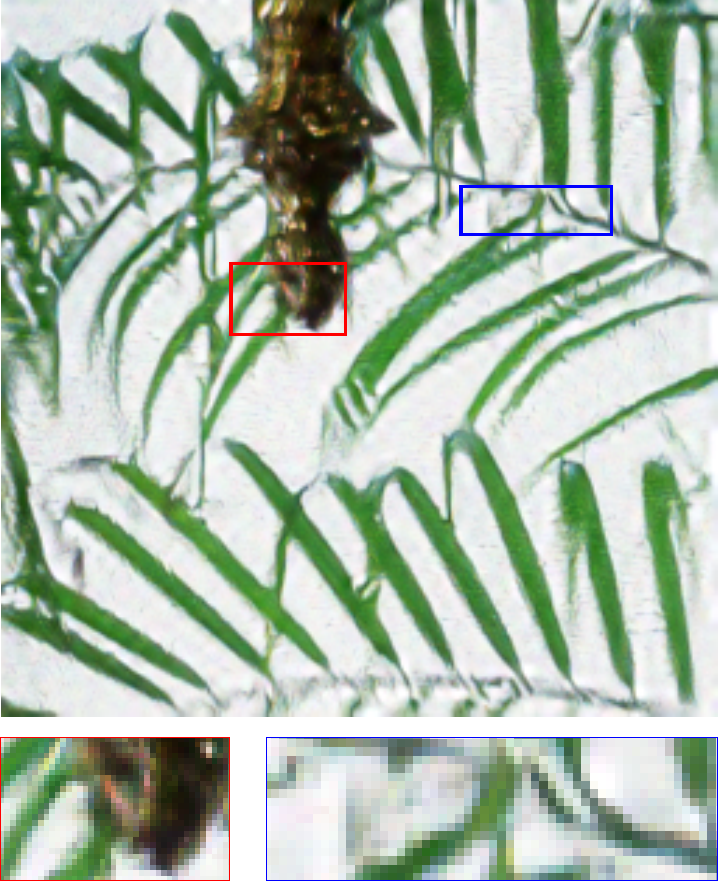} &
        \includegraphics[width=0.12\textwidth]{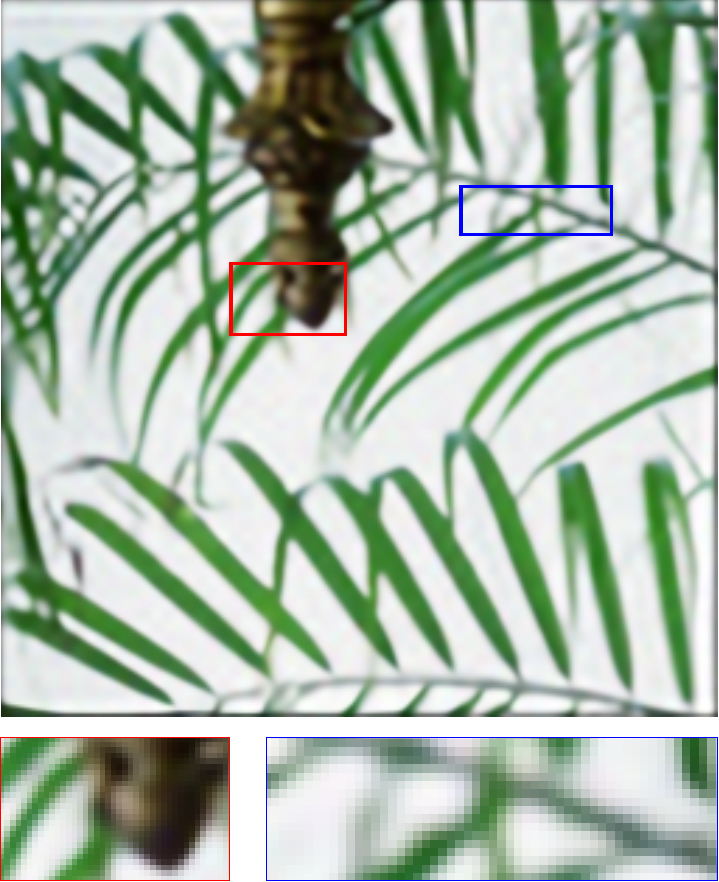} &
        \includegraphics[width=0.12\textwidth]{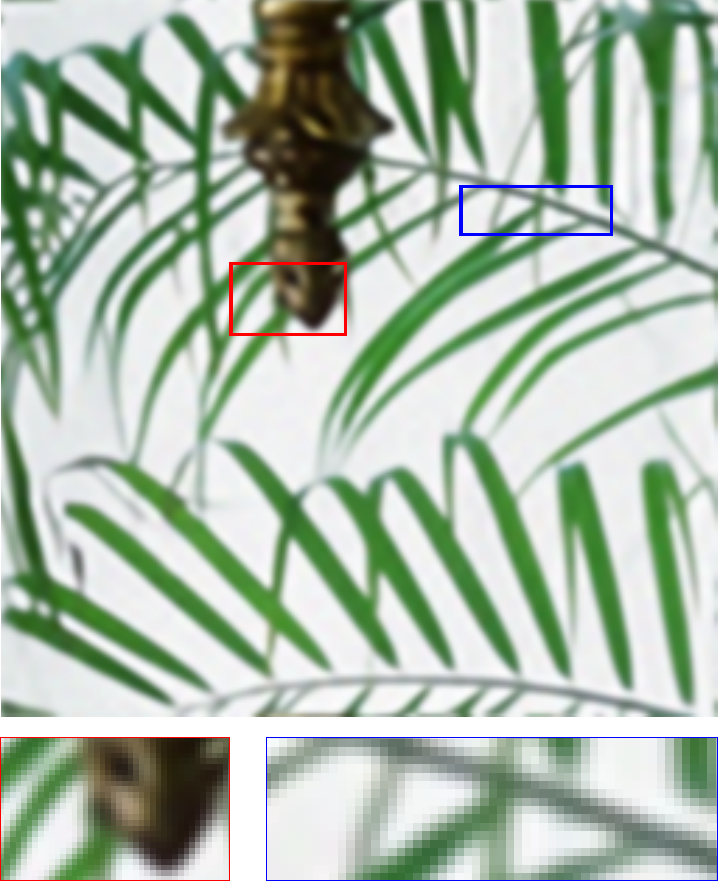} &
        \includegraphics[width=0.12\textwidth]{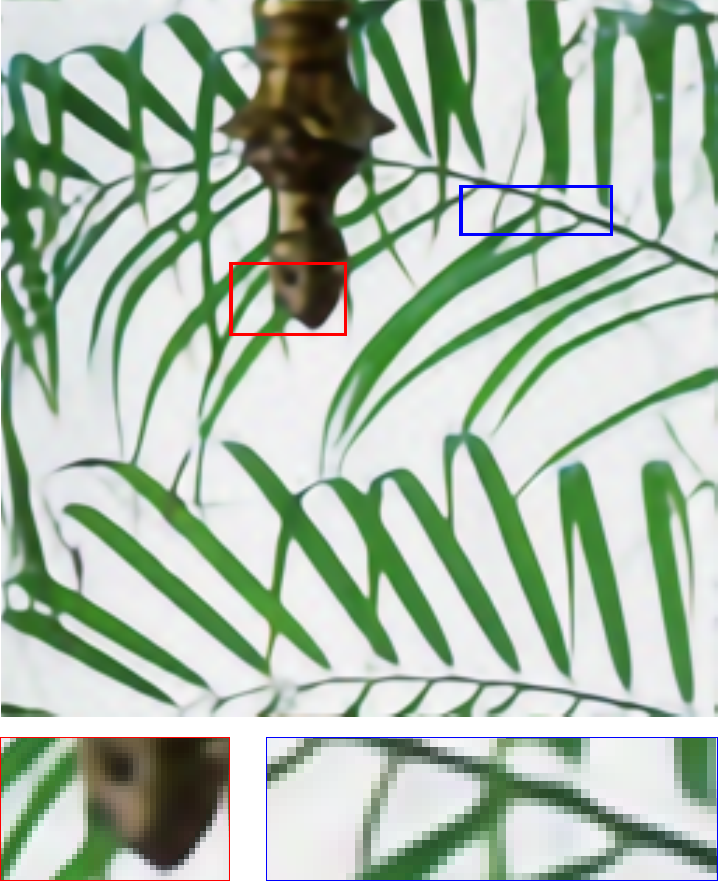} \\
        21.17 dB & 21.23 dB & 22.44 dB & 22.47 dB
    \end{tabular}
    \vspace{-1em}
    \caption{Reconstruction results with various algorithms for a SR$\times$4 problem. (d)-(h) show reconstructions obtained with \autoref{alg:sto_sgd} for various priors.}
    \label{fig:sr4_vanilla}
    \vspace{-1em}
\end{figure}

\begin{figure}[t]
\centering
\includegraphics[height=0.18\textwidth]{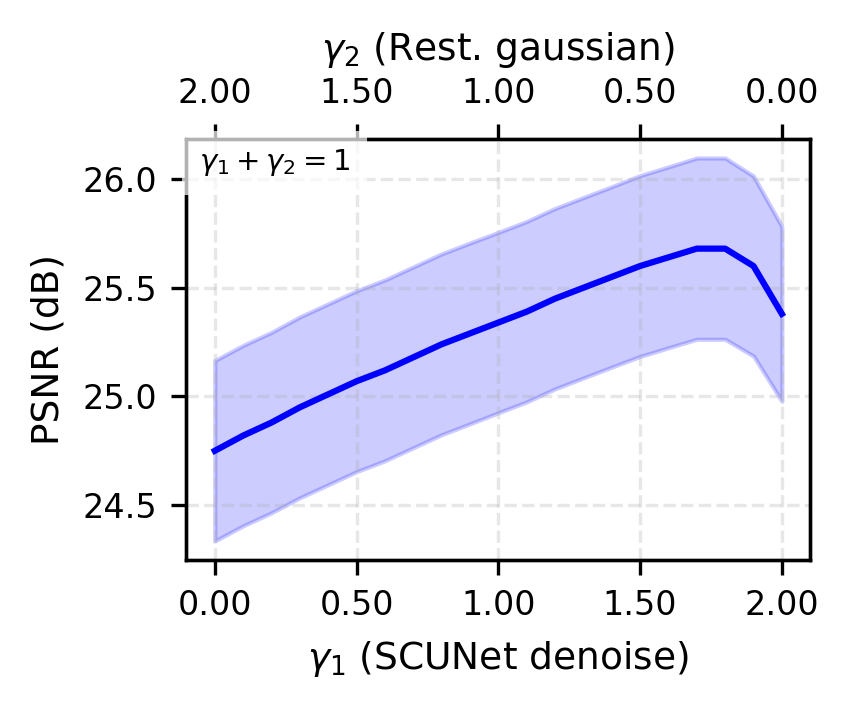}
\includegraphics[height=0.18\textwidth]{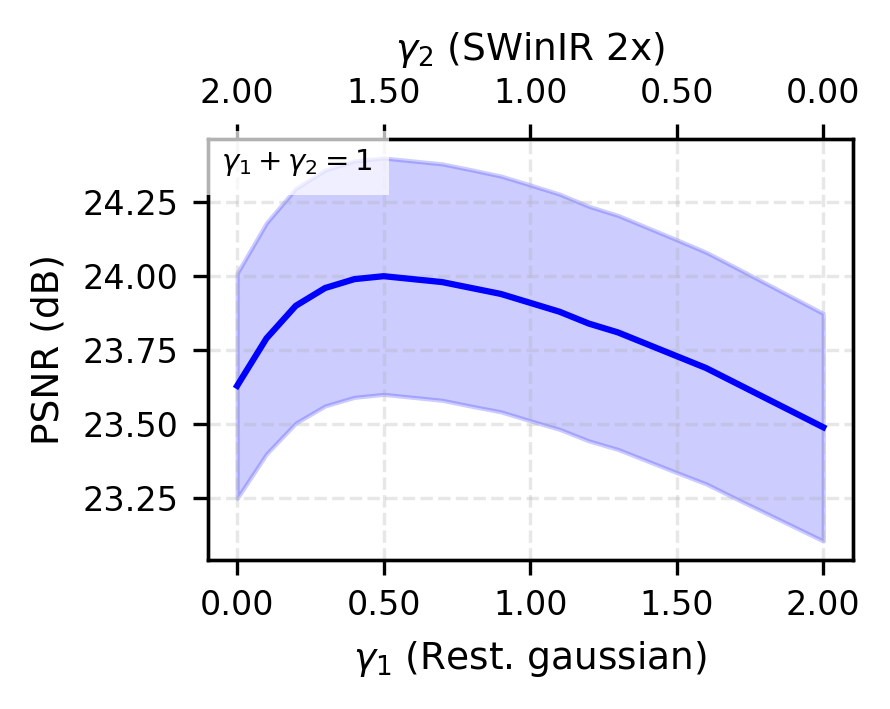}
\vspace{-1em}
\caption{$\gamma$ vs PSNR within the reconstruction quality for two different problems. Left: Gaussian deblurring, right: SR$\times$4. The $\gamma_1$ and $\gamma_2$ parameter control the strength of the associated prior.}
\label{fig:gamma_vs_psnr}
\vspace{-1em}
\end{figure}

\begin{figure*}[ht]
\small
\centering
\begin{tabular}{@{\hskip 0pt}c @{\hskip 3pt} c @{\hskip 3pt} c @{\hskip 3pt} c @{\hskip 3pt} c @{\hskip 3pt} c@{\hskip 0pt}}
\centering
Observed & DRP \cite{hu2023restoration} & DPIR \cite{zhang2021plug} & DiffPIR \cite{zhu2023denoising} & Proposed & Groundtruth \\
\includegraphics[width=0.16\textwidth]{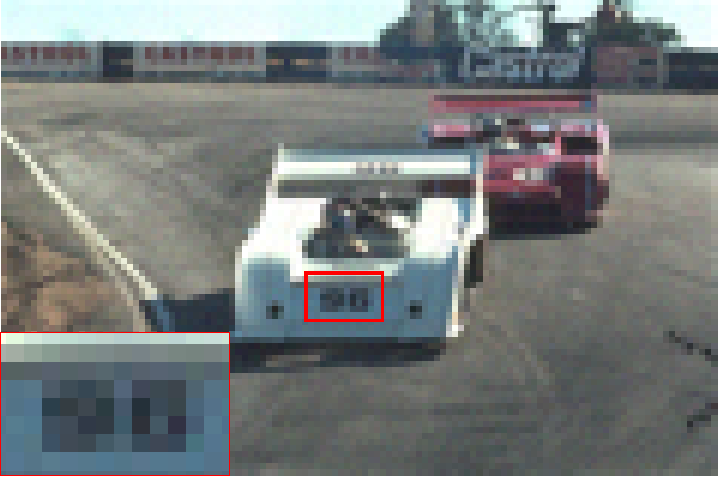} &
\includegraphics[width=0.16\textwidth]{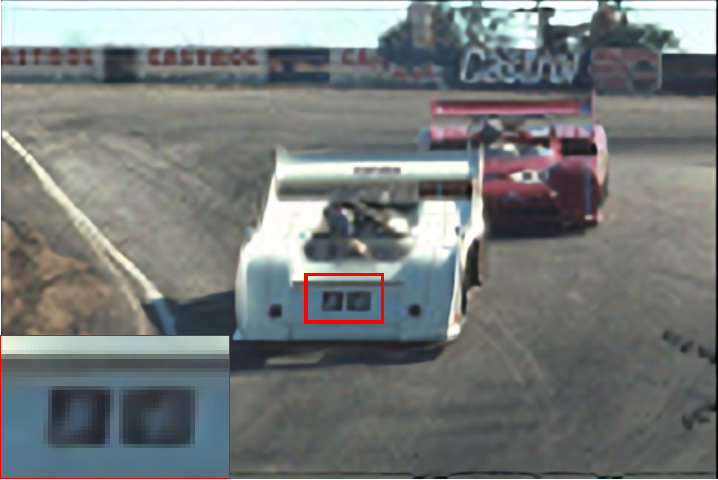} &
\includegraphics[width=0.16\textwidth]{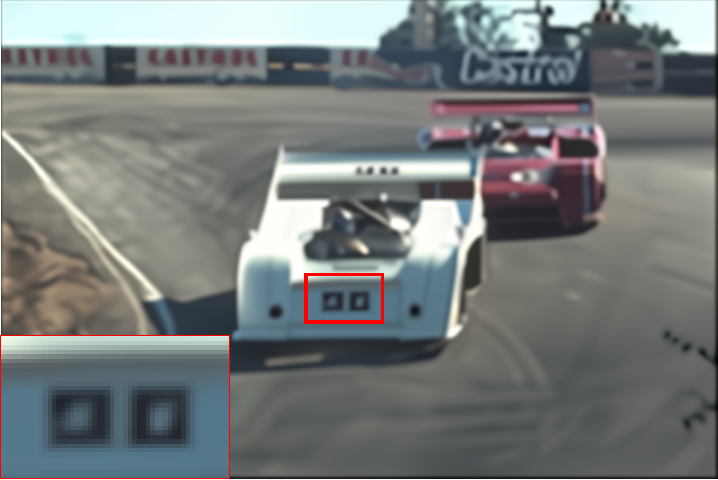} &
\includegraphics[width=0.16\textwidth]{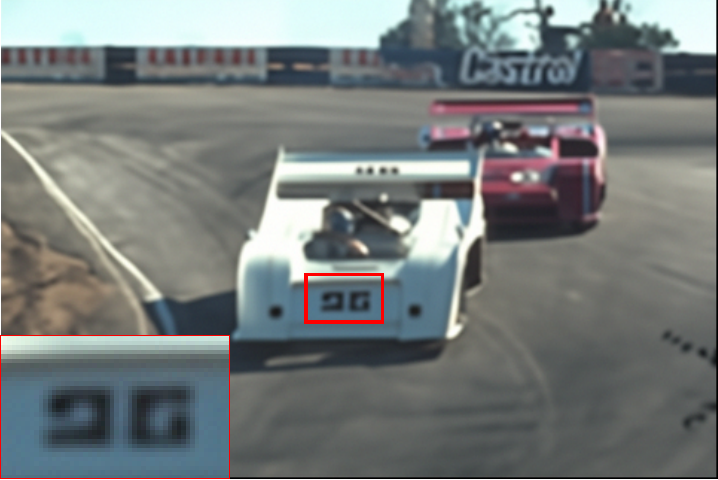} &
\includegraphics[width=0.16\textwidth]{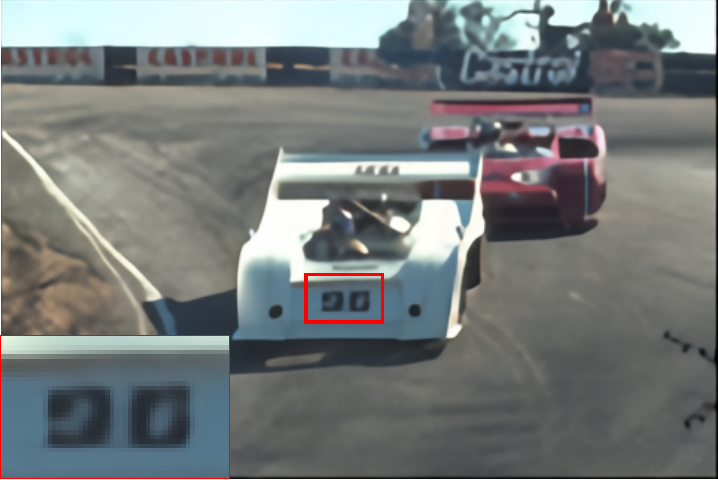} &
\includegraphics[width=0.16\textwidth]{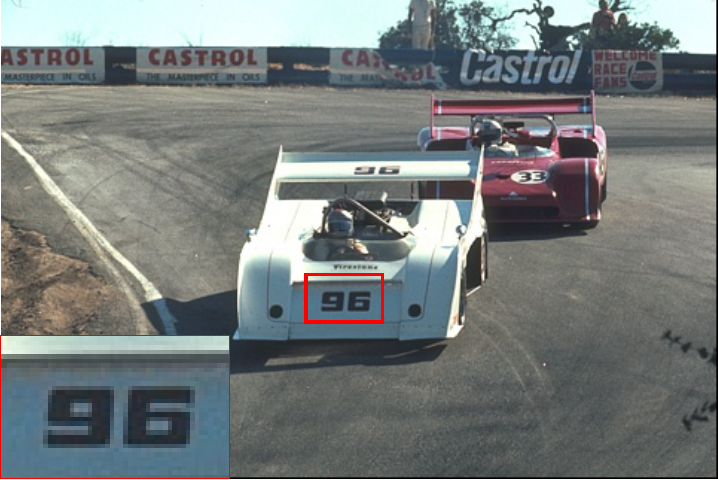}  \\
& (22.02, 0.52) & (24.02, 0.55) & (24.72, 0.49) & (25.11, 0.44) & (PSNR, LPIPS)\\
\includegraphics[width=0.16\textwidth]{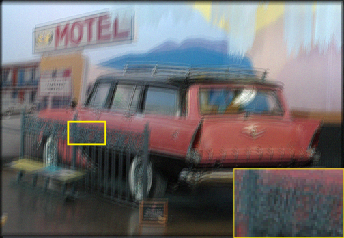} &
\includegraphics[width=0.16\textwidth]{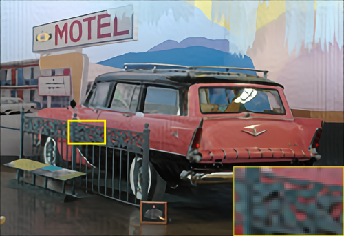} &
\includegraphics[width=0.16\textwidth]{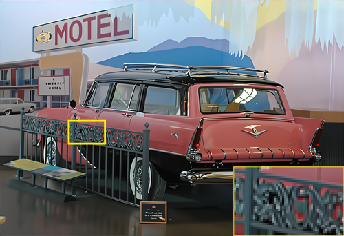} &
\includegraphics[width=0.16\textwidth]{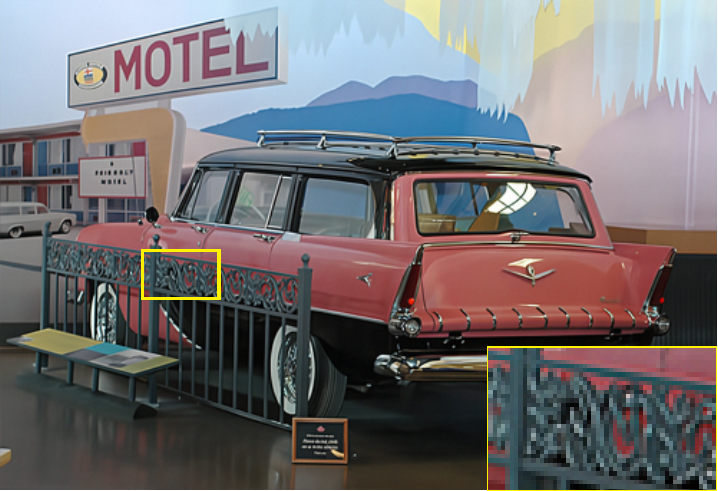} &
\includegraphics[width=0.16\textwidth]{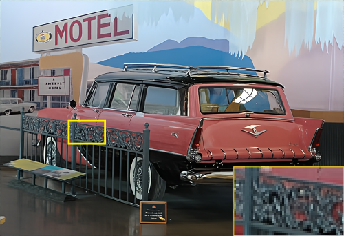} &
\includegraphics[width=0.16\textwidth]{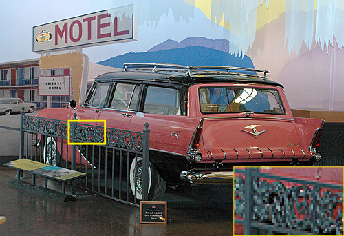} \\
 & (22.68, 0.27) & (29.19, 0.11) & (29.42, 0.09) & (28.64, 0.10) & (PSNR, LPIPS)\\
\includegraphics[width=0.16\textwidth]{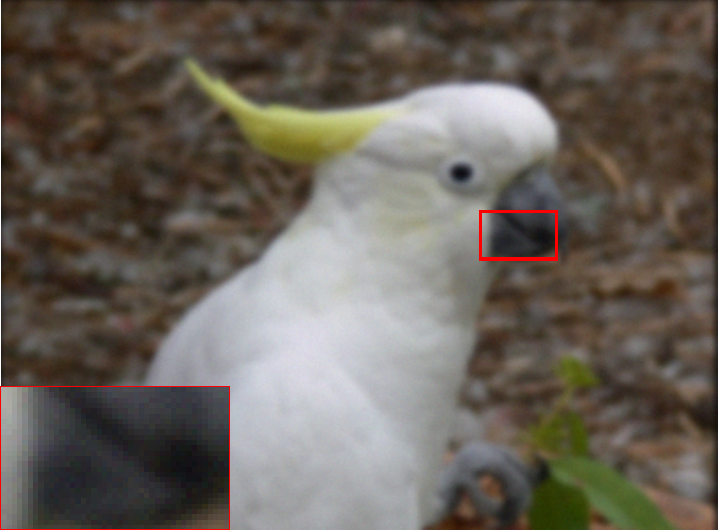} &
\includegraphics[width=0.16\textwidth]{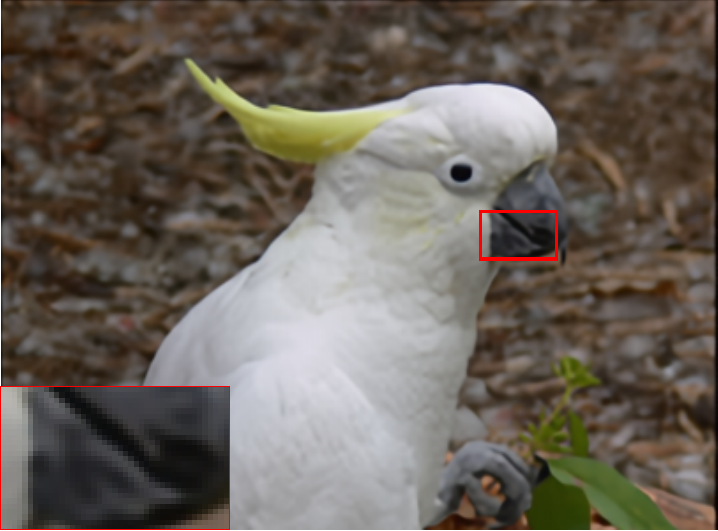} &
\includegraphics[width=0.16\textwidth]{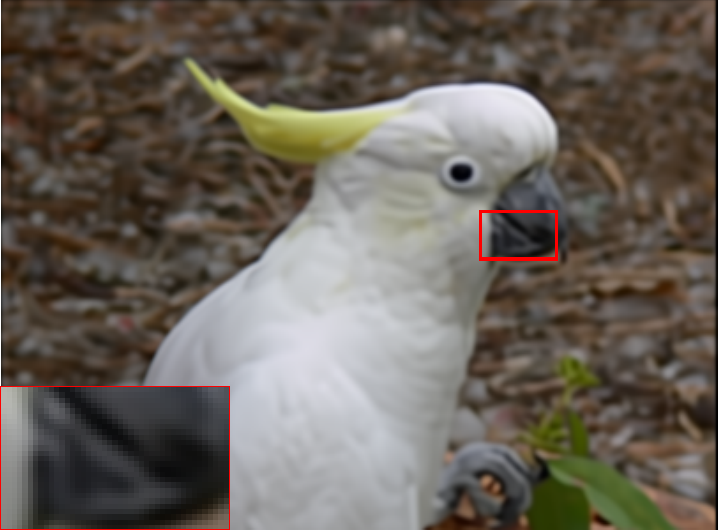} &
\includegraphics[width=0.16\textwidth]{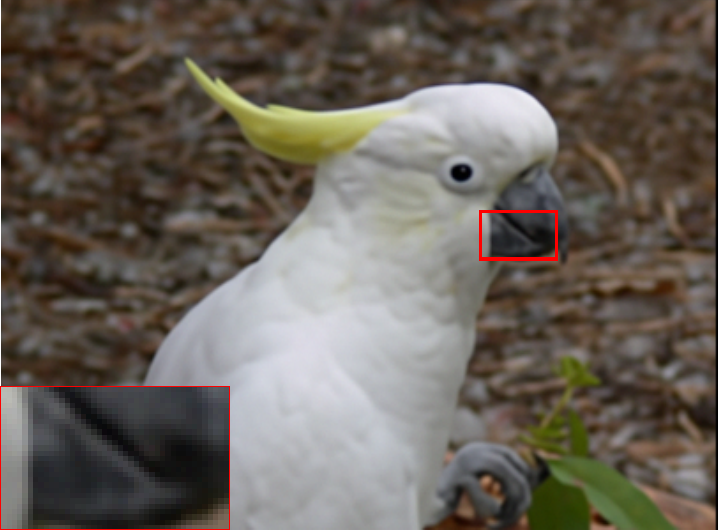} &
\includegraphics[width=0.16\textwidth]{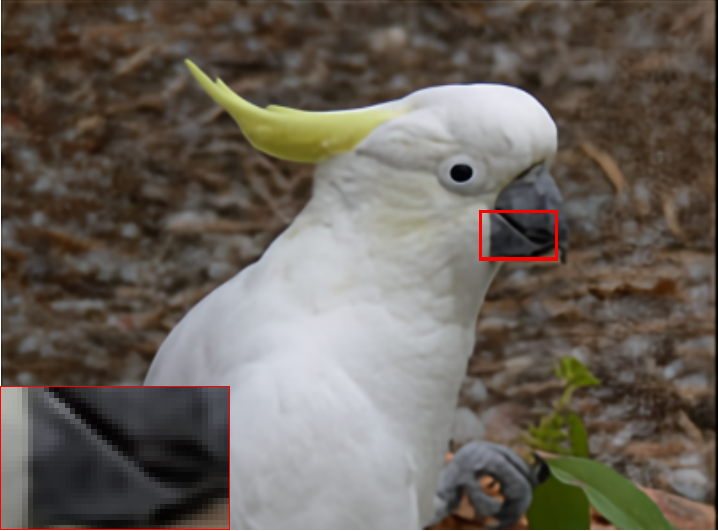} &
\includegraphics[width=0.16\textwidth]{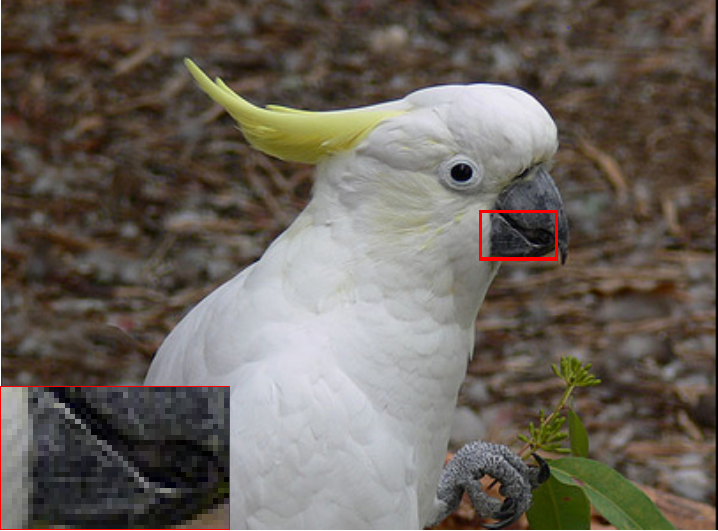} \\
& (30.22, 0.22) & (30.03, 0.32) & (30.42, 0.27) & (30.17, 0.31) & (PSNR, LPIPS)\\
\end{tabular}
\vspace{-1em}
\caption{Image restoration with various algorithms. Top: SR$\times 4$ problem with $\sigma = 0.01$ on BSD20. Middle: Motion blur on Imnet100. Bottom: Gaussian deblurring with blur kernel of size 3 and $\sigma = 0.01$ on Imnet100.}
\label{fig:gb_results}
\vspace{-1em}
\end{figure*}

\subsection{Combining multiple priors}

As discussed in \autoref{sect:proposed}, the set of fixed-points of $\operatorname{R}\circ D$ depends on both the restoration model $\operatorname{R}$ and its associated degradation operator $D$. These sets can vary significantly across different models, as illustrated in \autoref{fig:summary}~(c): the SCUNet JPEG prior yields fixed-points with piecewise-constant regions, while the deblurring model preserves more texture.
The formulation \eqref{eq:general_min_exp} naturally encompasses multiple priors, where each distance term $d_{C_\xi}$ is derived from a restoration model $\operatorname{R}_n$  associated to its training degradation tasks $\mathcal{D}^n$. The influence of each prior can be controlled through the parameter $\gamma_n$ in \autoref{alg:sto_sgd}: values close to 0 minimize the effect of restoration model $\operatorname{R}_n$, while larger values increase its contribution. This parameter plays a role analogous to regularization parameters in classical PnP algorithms \cite{hurault2024convergent}. For $N$ priors, we select $\gamma_n\in[0, 1]$ and ensure $\sum_{n=1}^N \gamma_n = 1$.

We demonstrate the benefits of combining multiple priors in \autoref{fig:gamma_vs_psnr}, where we study two model pairs: (Gaussian Restormer, SwinIR 2$\times$) and (Gaussian Restormer, SCUNet).
Fine-tuning the $\gamma_n$ parameters leads to consistent improvements in reconstruction quality, highlighting the advantage of our ensembling approach.
\autoref{tab:results_combined} provides quantitative results using an ensemble of three models: the blind SCUNet denoiser, SwinIR$\times$2, and Gaussian Restormer.
Our method outperforms the DRP algorithm, as further illustrated by the visual results in \autoref{fig:gb_results}.
Note that we return the iterate $u_k$ instead of $x_{k+1}$ in \autoref{alg:sto_sgd}, as we observed it yields superior reconstructions.

\subsection{Conditioning the prior on the measurements}

\begin{figure}[ht]
\small
    \centering
    \begin{tabular}{@{\hskip 0pt}c @{\hskip 3pt} c@{\hskip 3pt} c @{\hskip 0pt}}
        \centering
        (a) Observed & (b) Denoising prior & (c) Inpainting prior \\
        \includegraphics[width=0.14\textwidth]{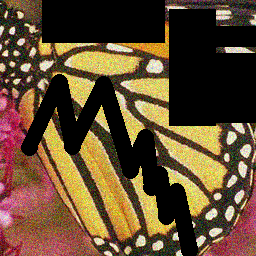} &
        \includegraphics[width=0.14\textwidth]{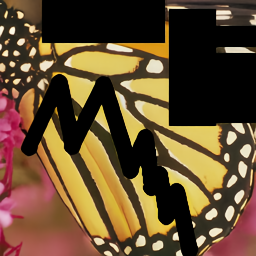} &
        \includegraphics[width=0.14\textwidth]{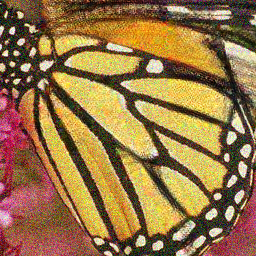} \\
        (d) DRP \cite{hu2023restoration} & (e) Denoise first & (f) Inpaint First \\
        \includegraphics[width=0.14\textwidth]{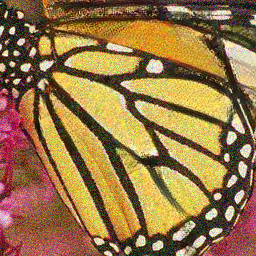} &
        \includegraphics[width=0.14\textwidth]{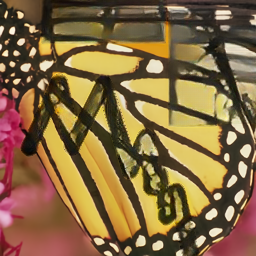} &
        \includegraphics[width=0.14\textwidth]{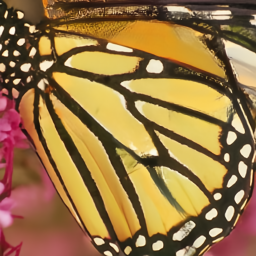} \\
        (g) ShaRP \cite{hu2024stochastic} & (h) DiffPIR & (i) {\bf Proposed} \\
        \includegraphics[width=0.14\textwidth]{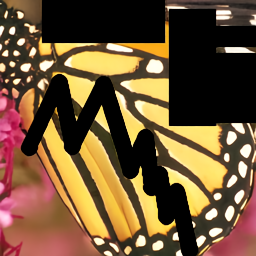} &
        \includegraphics[width=0.14\textwidth]{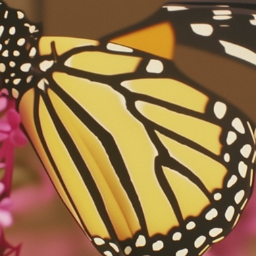} &
        \includegraphics[width=0.14\textwidth]{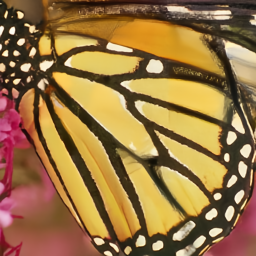} 
    \end{tabular}
    \vspace{-1em}
    \caption{Comparison of different restoration approaches for a noisy inpainting problem.}
    \label{fig:inpainting_comparison}
    \vspace{-2em}
\end{figure}

A natural extension of our framework involves incorporating the forward operator $A$ as one of the degradation operators $H$. Specifically, given a restoration model $\operatorname{R}_1$ trained to solve inverse problems of the form $y = \widetilde{A}x + w$ with operator $\widetilde{A}$ similar to $A$, we can condition one of the priors on the measurement process by setting $H_1 = A$ in \autoref{alg:sto_sgd}.

To showcase this strategy, we consider a noisy image inpainting problem, where $A$ is a binary masking operator. We implement \autoref{alg:sto_sgd} with two complementary priors: an inpainting-specific model (LAMA) as $\operatorname{R}_1$ with $H_1 = A$, and a general denoising model (SCUNet) as $\operatorname{R}_2$. Visual results in \autoref{fig:inpainting_comparison} show that while neither $\operatorname{R}_1$ nor $\operatorname{R}_2$ alone succeeds in our FiRe framework (panels (b) and (c)), their combination successfully addresses both the inpainting and denoising aspects of the problem (panel (i)).

We compare against several baselines. The DRP algorithm, despite being theoretically ill-defined due to the non-invertibility of $A$, can be implemented using pseudo-inverse computations. Using a LAMA inpainting model, DRP effectively handles missing pixels but fails to denoise the image. Conversely, ShaRP is effective at denoising but fails to properly inpaint masked pixels due to the $H_1^\top H_1$ term in its gradient. We also evaluate naive sequential approaches combining denoising and inpainting models. These reveal a strong dependency on the ordering of operations, a limitation that FiRe overcomes through parallel processing of different models.

\section{Conclusion}

We have introduced Fixed-points of Restoration (FiRe) priors, extending the implicit prior paradigm beyond denoising to general restoration models. Our approach leverages the observation that natural images are fixed points of restoration models composed with their training degradation operators. This yields both an explicit prior formulation and a versatile algorithm requiring no restrictive assumptions on the degradation prior. Experiments across various restoration tasks demonstrate two key capabilities: the combination of multiple restoration models in an ensemble-like fashion, and the incorporation of measurement-aware priors. These features enable FiRe to address challenging scenarios where traditional PnP methods struggle, such as noisy inpainting, all within a unified framework.
A challenge remains with the tuning of hyperparameters, which can significantly impact performance across different degradation settings. Addressing this limitation through adaptive or learned parameter selection would be a promising direction for future research.

\section*{Acknowledgements}
This work was supported by the BrAIN grant (ANR-20-CHIA-0016) and was granted access to the HPC resources of IDRIS under the allocation 2023-AD011014344 made by GENCI. Ulugbek Kamilov was supported by the NSF CAREER award under CCF-2043134.

{
    \small
    \bibliographystyle{ieeenat_fullname}
    \bibliography{main}
}

\clearpage
\setcounter{page}{1}
\maketitlesupplementary

\appendix

\section{Convergence of stochastic algorithm}

We now formalize the convergence result stated in \autoref{sect:expected_prior} for algorithm \eqref{eq:grad_sto}. Our analysis relies on two key assumptions. First, we require unbiased gradient estimators with variance growth condition:

\begin{assumption}
\label{ass:sto_1}
For all $k$, $g_k$ in \eqref{eq:grad_sto} is an unbiased estimator of $\frac{1}{2}\nabla d_C^2$. More precisely, assuming that the sequence $(x_k)_{k\in\mathbb{N}}$ is adapted to   the filtration $\{\mathcal{F}_k\}_{k\geq 0}$, we assume that
$\mathbb{E}_{A_k, e_k}[g_k|\mathcal{F}_k] = \frac{1}{2}\nabla d_C^2(x_k)$ where $d_C$ is defined as in \autoref{prop:grad}. Furthermore, we assume that there exists constants $A, B\geq 0$ such that
\begin{equation}
    \mathbb{E}[\|g_k-\frac{1}{2}\nabla d_C^2(x_k)\|^2|\mathcal{F}_k] \leq A(d_C^2(x_k))+B
\end{equation}
holds almost surely for all $k\in \mathbb{N}$.
\end{assumption}

\noindent Second, we require a standard assumption on the step sizes:

\begin{assumption}
\label{ass:sto_2}
The stepsize $(\gamma_k)_{k\in\mathbb{N}}$ satisfies $\sum_{k=0}^\infty \gamma_k=\infty$ and $\sum_{k=0}^\infty \gamma_k^2<\infty$.
\end{assumption}

\noindent Under those assumptions, the proposed algorithm rewrites as a proximal stochastic gradient algorithm, and we can derive the following result.

\begin{proposition}
Assume that Assumptions~\ref{ass:sto_1} and \ref{ass:sto_2} hold, 
and define the residual function 
\begin{equation}
F(x) = x-\operatorname{prox}_{\lambda f}(x-\frac{1}{2}d_C^2(x)). 
\end{equation}
Then we have that $\mathbb{E}[F(x_k)]\underset{k\to\infty}{\longrightarrow} 0$.
\end{proposition}

\begin{proof}
First, we have that in our setting, $f$ is convex. Moreover, we have that $d_C^2$ is bounded from below, and from \autoref{prop:grad}, we have that $\nabla d_C^2$ is Lipschitz. The result then follows from \cite[Theorem Corollary 3.6]{li2022unified}.
\end{proof}

\section{Prior loss}

To investigate the underlying prior associated with the different models, we investigate the quantity $d(y) = \|y - R(Hy + w)\|$ where $y$ are images with various levels of degradation. More precisely, we set $y = \sigma_{\text{blur}}*x + \sigma_{\text{noise}}n$ where $n\sim\mathcal{N}(0, \operatorname{Id})$; thus, as $\sigma_\text{noise}\to0$ and $\sigma_\text{blur}\to0$, $y$ tends to a natural image.
We plot values of $d$ obtained for different models $\operatorname{R}$ in \autoref{fig:fig_prior}.
We observe that the lowest values of $d(y)$ are obtained for noiseless, smooth images, suggesting that the proposed priors tends to show a smoothing property and promotes image regularity.

\begin{figure}[t]
\small
    \centering
    \includegraphics[width=0.48\textwidth]{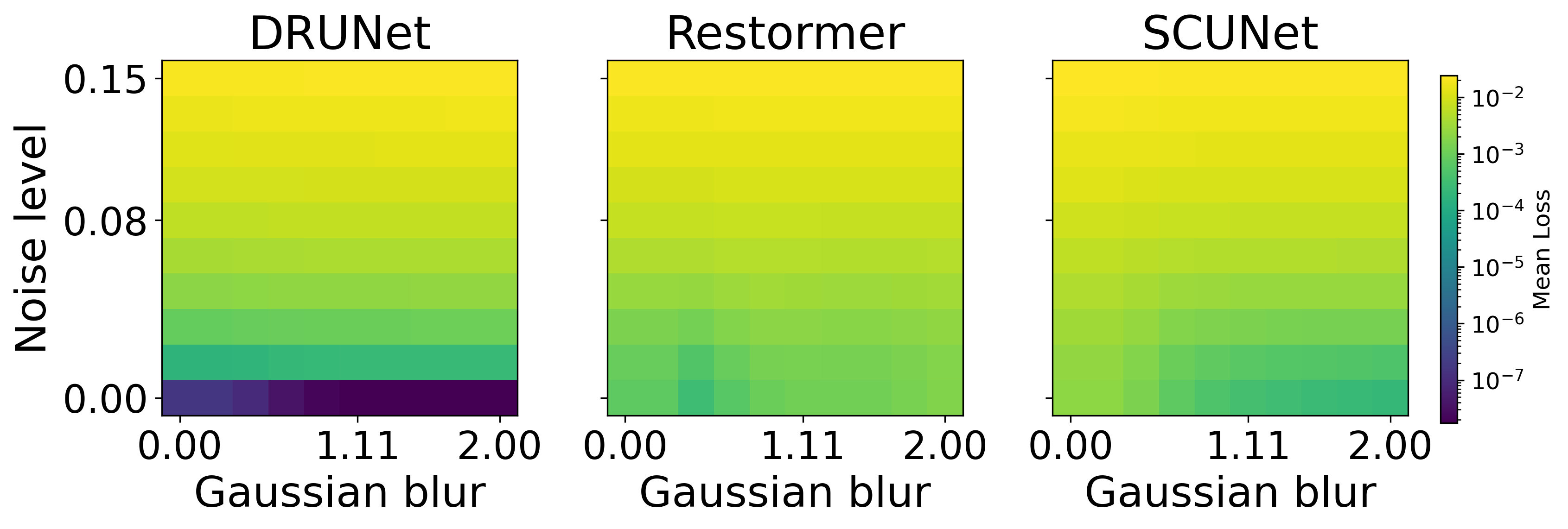}
    \vspace{-2.5em}
    \caption{Average distance $\|y - R(H(y))\|_2$ for different models and different degradations for $y$. More precisely, $y$ is a blurred and noisy version of $x$ defined as $y= \sigma_{\text{blur}} * x + \sigma_{\text{noise}} n$.}
    \label{fig:fig_prior}
\vspace{-1.5em}
\end{figure}

\section{Finetuning of restoration priors}

In this section, we summarize the implementation details for the finetuned models used in our experiments. We consider 3 finetuned models: two versions of the Restormer model for image deblurring (Gaussian and motion), and random mask inpainting. All finetunings are performed on the training dataset from \cite{zhang2021plug}. 

\vspace{-0.5em}
\subsection{Deblurring models}

The Restormer model from \cite{zamir2022restormer} is trained on real blur images with blur kernels difficult to simulate. In turn, computing the degradation set $\mathcal{D}$ associated to the model is not straightforward. Instead, we propose to finetune the model in two easily simulated setups.

\noindent \textbf{Restormer Gaussian:} Finetuned on Gaussian and diffraction blur removal with kernel size sampled uniformly at random $\sigma_\text{blur} \in [0.001, 4]$ and additive Gaussian noise with standard deviation sampled uniformly at random $\sigma \in [0.001, 0.1]$. We train on randomly cropped image patches of sizes 256$^2$ with batch size 8. Optimization is performed on L1 loss and uses Adam optimizer with default PyTorch parameters and learning rate 1e-4 for 90k steps. 

\noindent \textbf{Restormer Motion:} Finetuned on motion blur removal from \cite{tachella2023deepinv} with trajectory length scale 0.6, Gaussian Process standard deviation 1.0, and additive Gaussian noise sampled uniformly at random $\sigma \in [0.001, 0.1]$. Other training parameters are the same as above.

\subsection{Inpainting models}

The LAMA model \cite{suvorov2022resolution} is trained on large blur kernels and performs well on large inpainting masks, but we observed suboptimal performance on binary random masks.

\noindent \textbf{LAMA random inpainting:} Finetuned for random inpainting with mask probability sampled uniformly at random $p \in [0.1, 0.9]$ (no noise is added during training). Only the last 4 convolutional layers are updated during training. We train on randomly cropped image patches of sizes 128$^2$ with batch size 64. Optimization is performed on L1 loss uses Adam optimizer with learning rate 1e-4 for 300k steps.

\clearpage

\begin{figure*}[!htbp]
\small
\centering
\begin{tabular}{@{\hskip 0pt}c @{\hskip 3pt} c @{\hskip 3pt} c @{\hskip 3pt} c @{\hskip 3pt} c @{\hskip 3pt} c@{\hskip 0pt}}
\centering
Observed & DRP \cite{hu2023restoration} & DPIR \cite{zhang2021plug} & DiffPIR \cite{zhu2023denoising} & Proposed & Groundtruth \\
\includegraphics[width=0.16\textwidth]{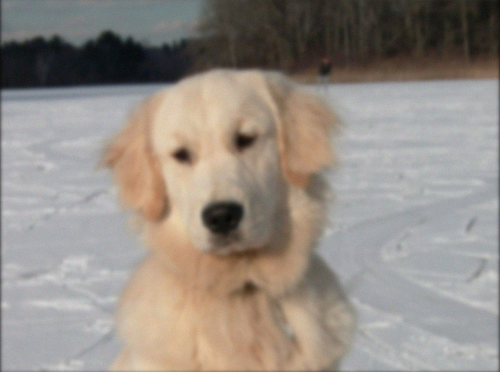} &
\includegraphics[width=0.16\textwidth]{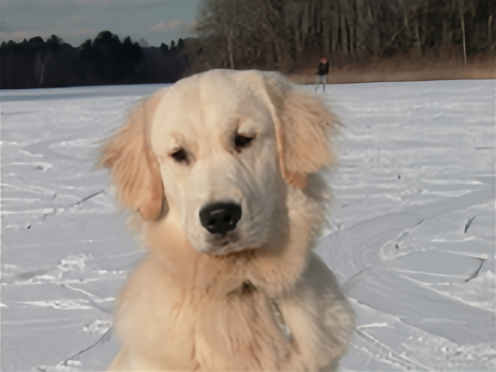} &
\includegraphics[width=0.16\textwidth]{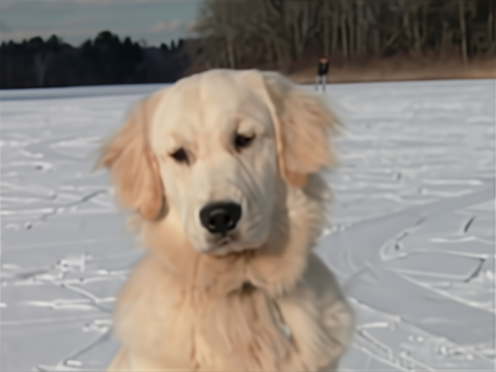} &
\includegraphics[width=0.16\textwidth]{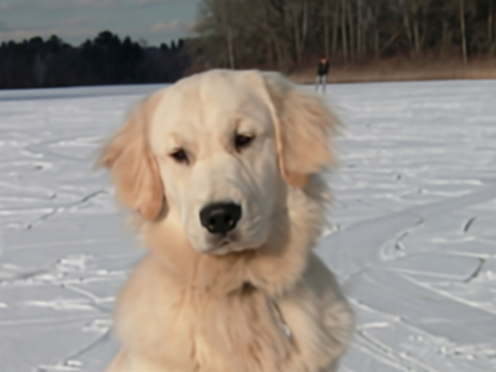} &
\includegraphics[width=0.16\textwidth]{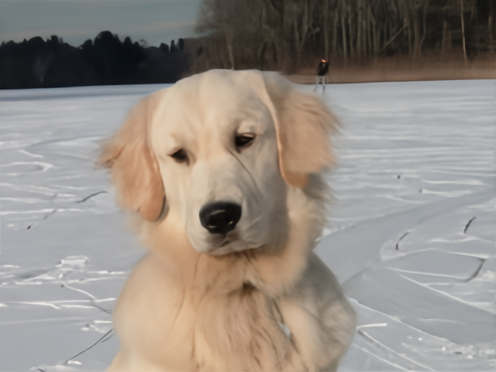} &
\includegraphics[width=0.16\textwidth]{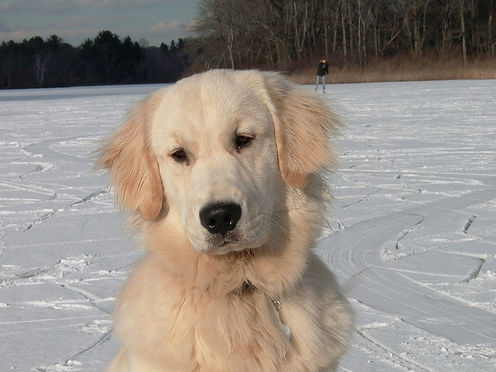}  \\
\includegraphics[width=0.16\textwidth]{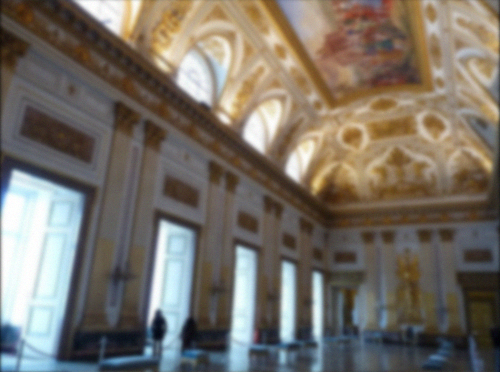} &
\includegraphics[width=0.16\textwidth]{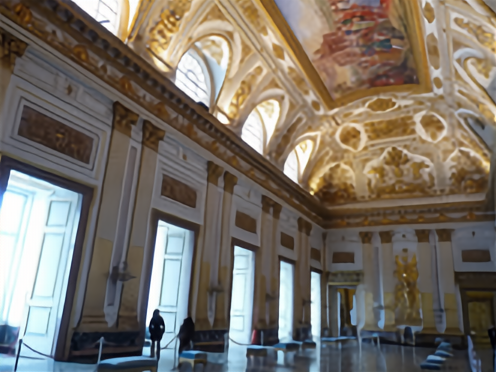} &
\includegraphics[width=0.16\textwidth]{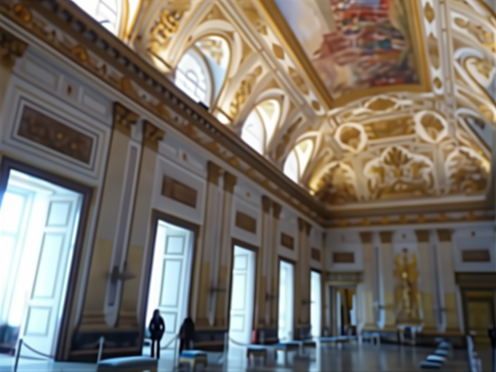} &
\includegraphics[width=0.16\textwidth]{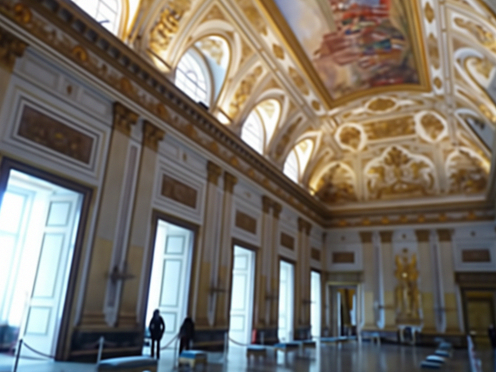} &
\includegraphics[width=0.16\textwidth]{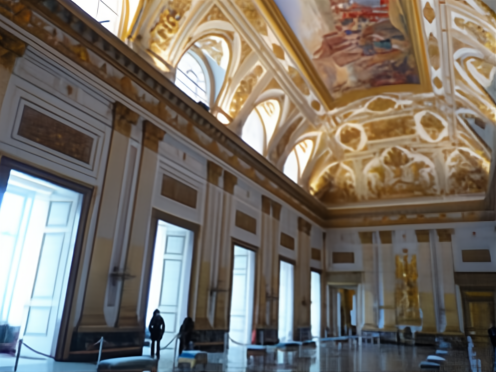} &
\includegraphics[width=0.16\textwidth]{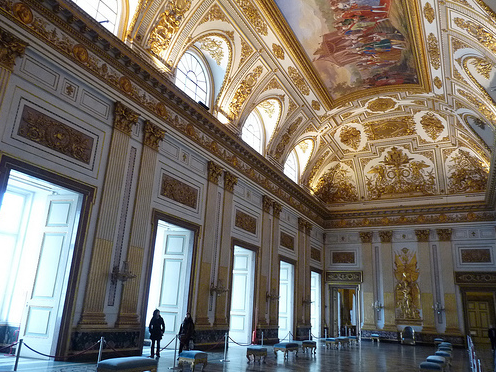}  \\
\includegraphics[width=0.16\textwidth]{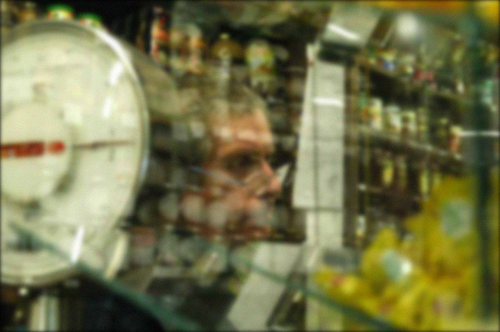} &
\includegraphics[width=0.16\textwidth]{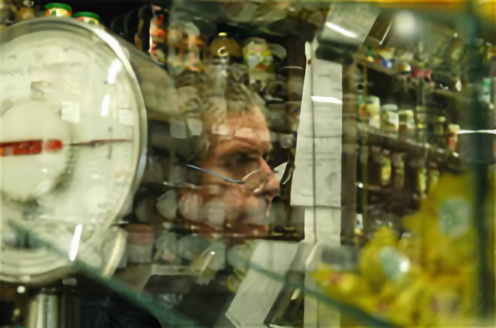} &
\includegraphics[width=0.16\textwidth]{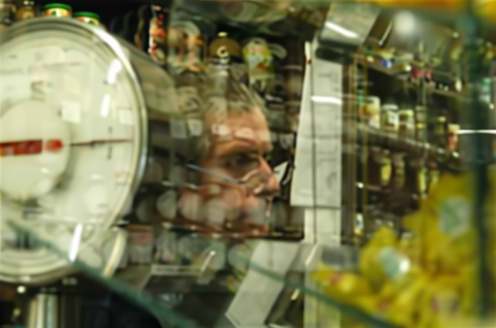} &
\includegraphics[width=0.16\textwidth]{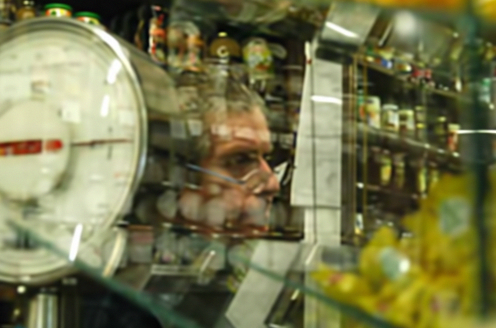} &
\includegraphics[width=0.16\textwidth]{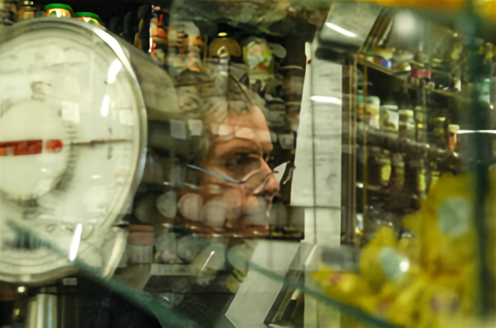} &
\includegraphics[width=0.16\textwidth]{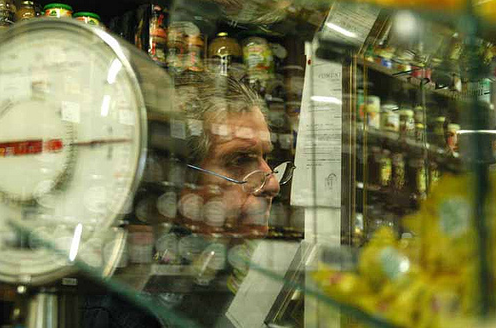}  \\
\includegraphics[width=0.16\textwidth]{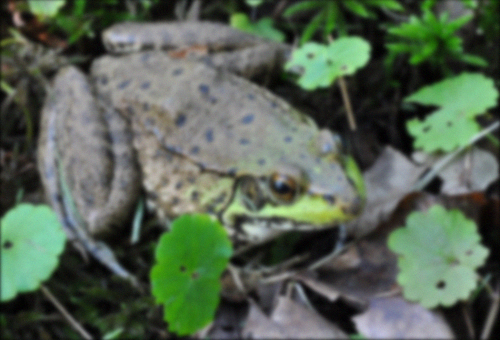} &
\includegraphics[width=0.16\textwidth]{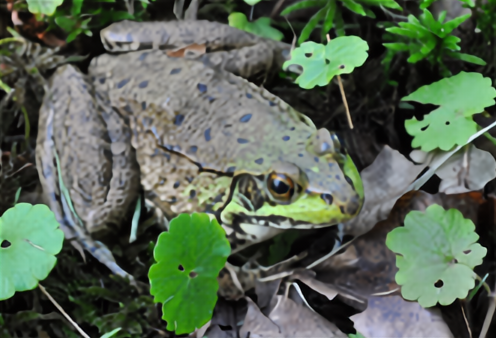} &
\includegraphics[width=0.16\textwidth]{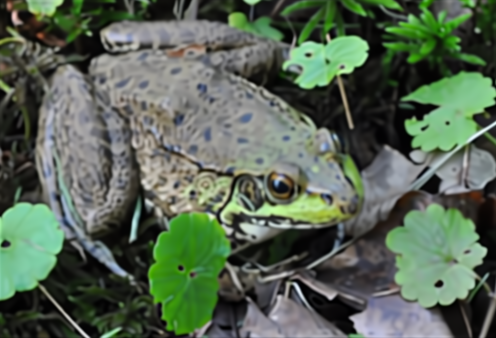} &
\includegraphics[width=0.16\textwidth]{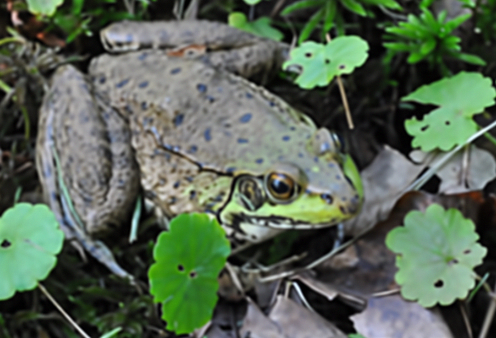} &
\includegraphics[width=0.16\textwidth]{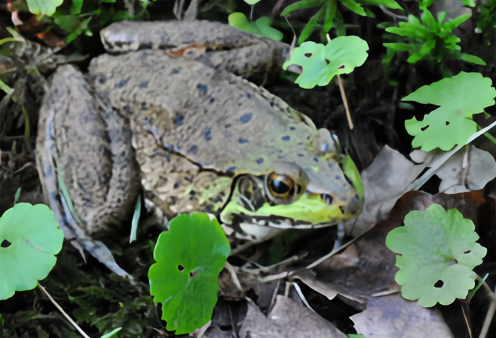} &
\includegraphics[width=0.16\textwidth]{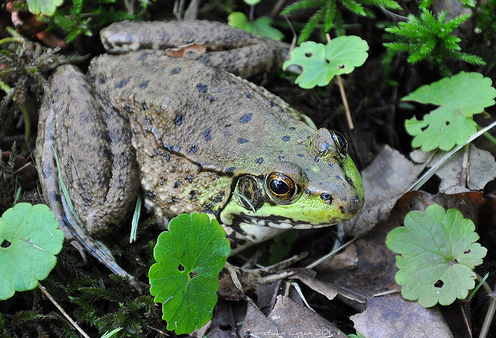}  \\
\includegraphics[width=0.16\textwidth]{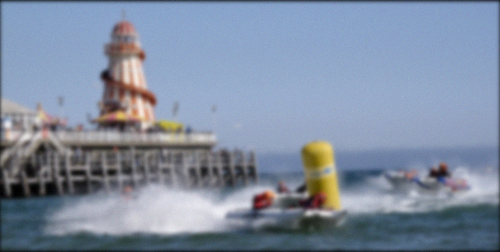} &
\includegraphics[width=0.16\textwidth]{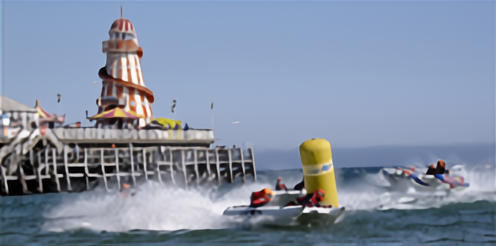} &
\includegraphics[width=0.16\textwidth]{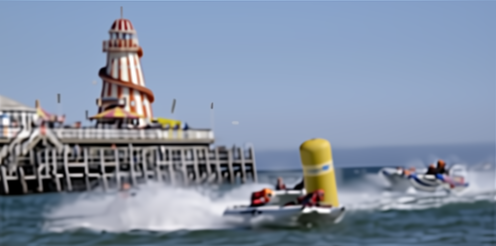} &
\includegraphics[width=0.16\textwidth]{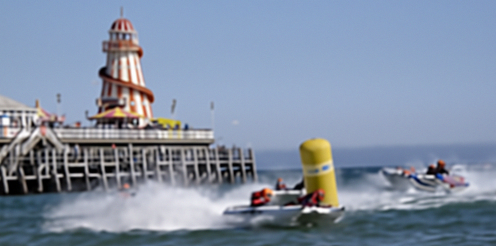} &
\includegraphics[width=0.16\textwidth]{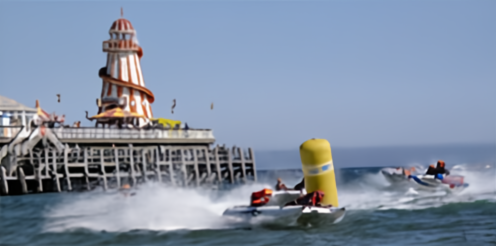} &
\includegraphics[width=0.16\textwidth]{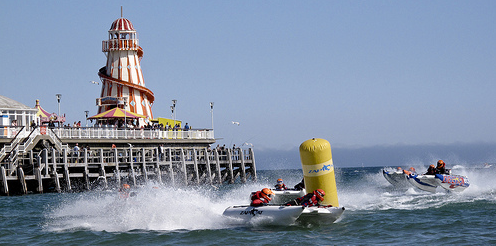}  \\
\end{tabular}
\caption{Gaussian deblurring on Imnet100, $\sigma = 0.01$.}
\label{fig:gaussian_sup}
\end{figure*}

\begin{figure*}[!htbp]
\small
\centering
\begin{tabular}{@{\hskip 0pt}c @{\hskip 3pt} c @{\hskip 3pt} c @{\hskip 3pt} c @{\hskip 3pt} c @{\hskip 3pt} c@{\hskip 0pt}}
\centering
Observed & DRP \cite{hu2023restoration} & DPIR \cite{zhang2021plug} & DiffPIR \cite{zhu2023denoising} & Proposed & Groundtruth \\
\includegraphics[width=0.16\textwidth]{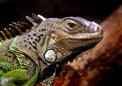} &
\includegraphics[width=0.16\textwidth]{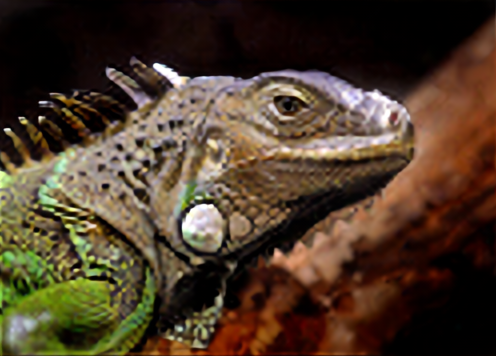} &
\includegraphics[width=0.16\textwidth]{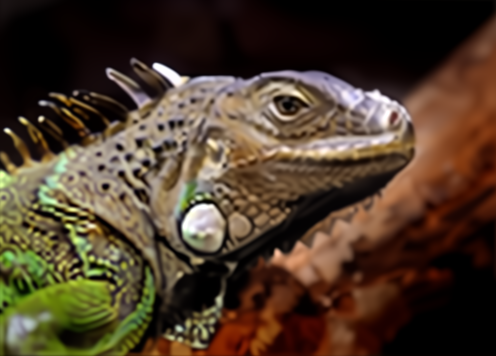} &
\includegraphics[width=0.16\textwidth]{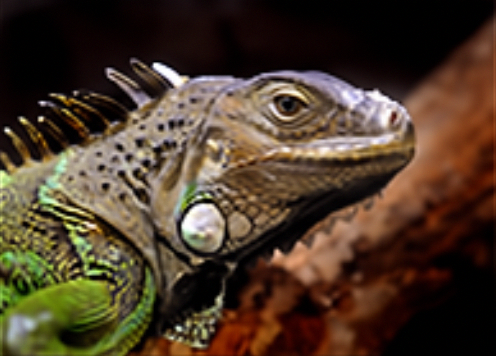} &
\includegraphics[width=0.16\textwidth]{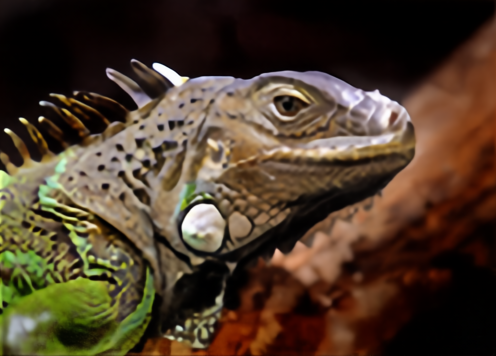} &
\includegraphics[width=0.16\textwidth]{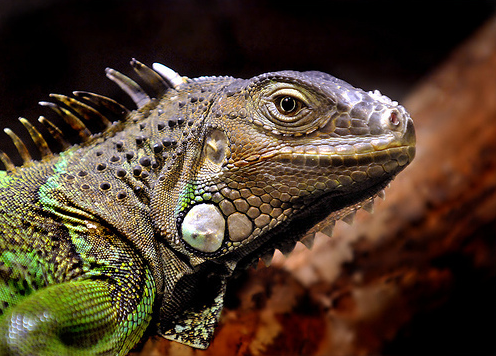}  \\
\includegraphics[width=0.16\textwidth]{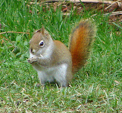} &
\includegraphics[width=0.16\textwidth]{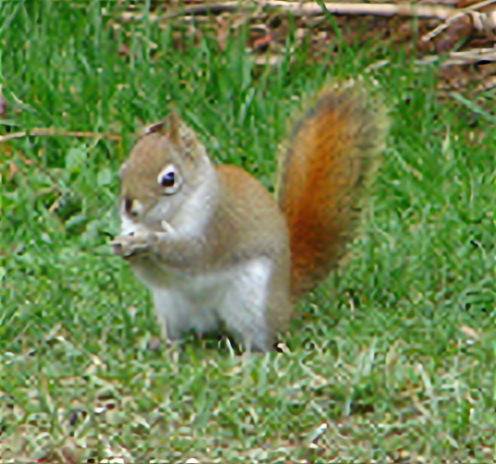} &
\includegraphics[width=0.16\textwidth]{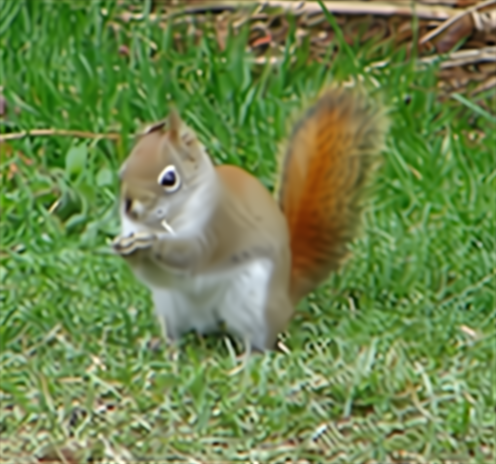} &
\includegraphics[width=0.16\textwidth]{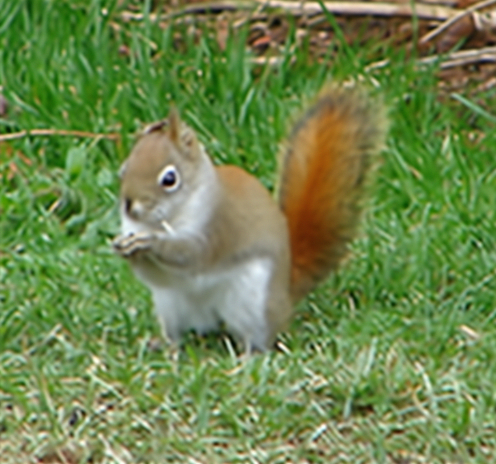} &
\includegraphics[width=0.16\textwidth]{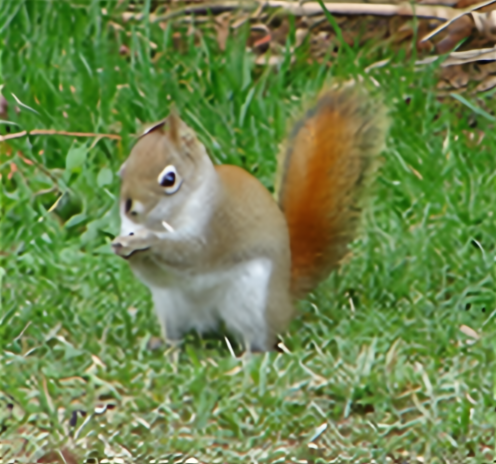} &
\includegraphics[width=0.16\textwidth]{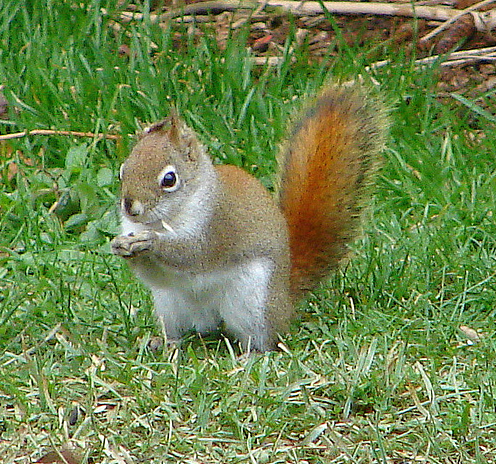}  \\
\includegraphics[width=0.16\textwidth]{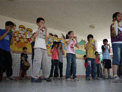} &
\includegraphics[width=0.16\textwidth]{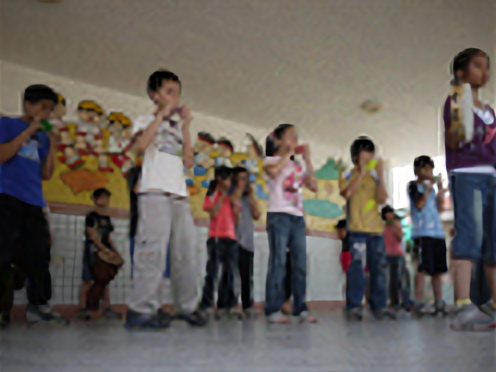} &
\includegraphics[width=0.16\textwidth]{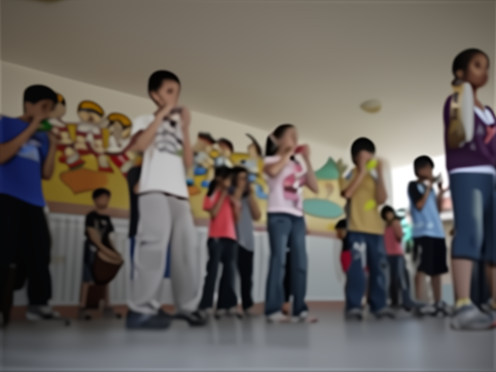} &
\includegraphics[width=0.16\textwidth]{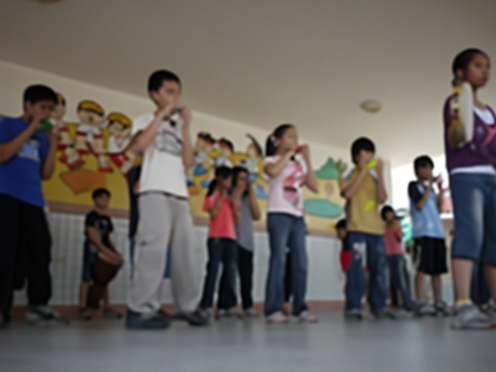} &
\includegraphics[width=0.16\textwidth]{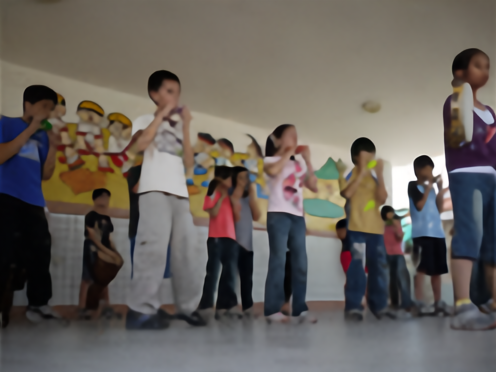} &
\includegraphics[width=0.16\textwidth]{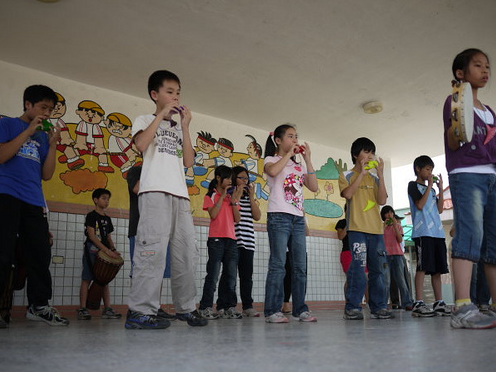}  \\
\includegraphics[width=0.16\textwidth]{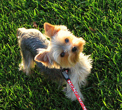} &
\includegraphics[width=0.16\textwidth]{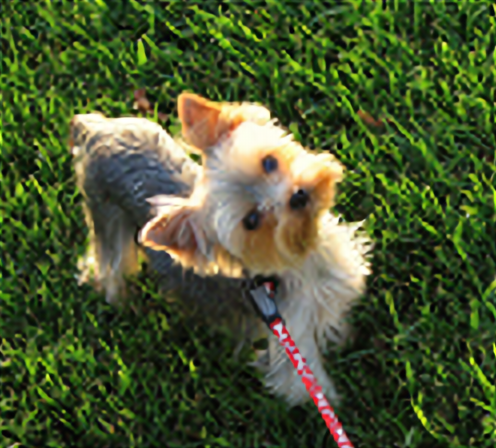} &
\includegraphics[width=0.16\textwidth]{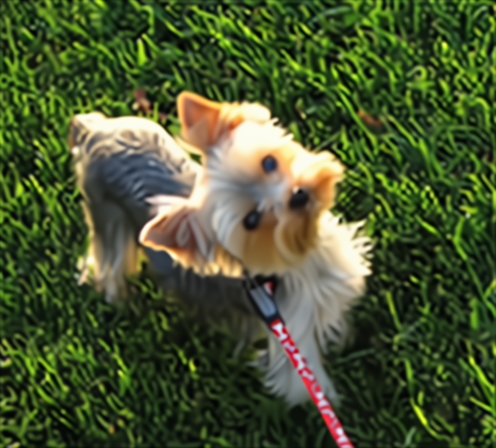} &
\includegraphics[width=0.16\textwidth]{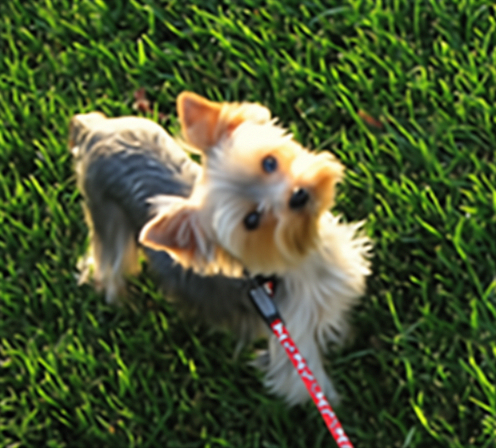} &
\includegraphics[width=0.16\textwidth]{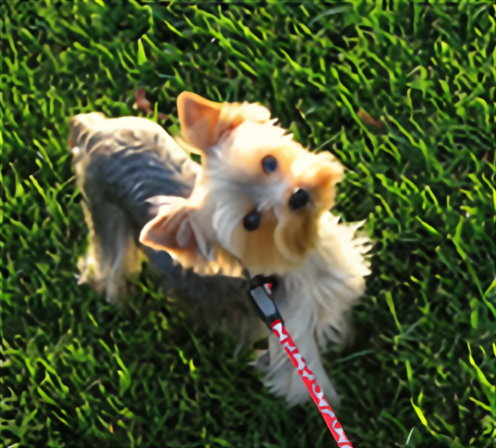} &
\includegraphics[width=0.16\textwidth]{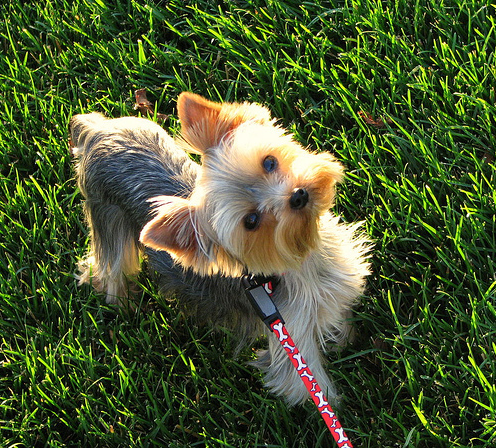}  \\
\end{tabular}
\caption{SRx4 on Imnet100, $\sigma = 0.01$.}
\label{fig:sr_sup}
\end{figure*}
\clearpage

\section{Additional visual results}

We provide further reconstruction results on the Gaussian deblurring problem in \autoref{fig:gaussian_sup} and on the SRx4 problem in \autoref{fig:sr_sup}.

\section{Influence of the degradation}

The choice of degradation operator significantly impacts reconstruction quality.
In~\cref{fig:ablation_study}, we give reconstruction metrics for varying degradation strengths using either a Gaussian Restormer prior and a DRUNet denoising prior.
This shows that a minimum degradation is required to stabilize the reconstruction, while increasing beyond a certain threshold leads to excessive smoothing.

\begin{figure}[h]
    \centering
    \begin{minipage}{0.18\linewidth}
        \centering
        \vspace{-1em}
        \includegraphics[width=\linewidth]{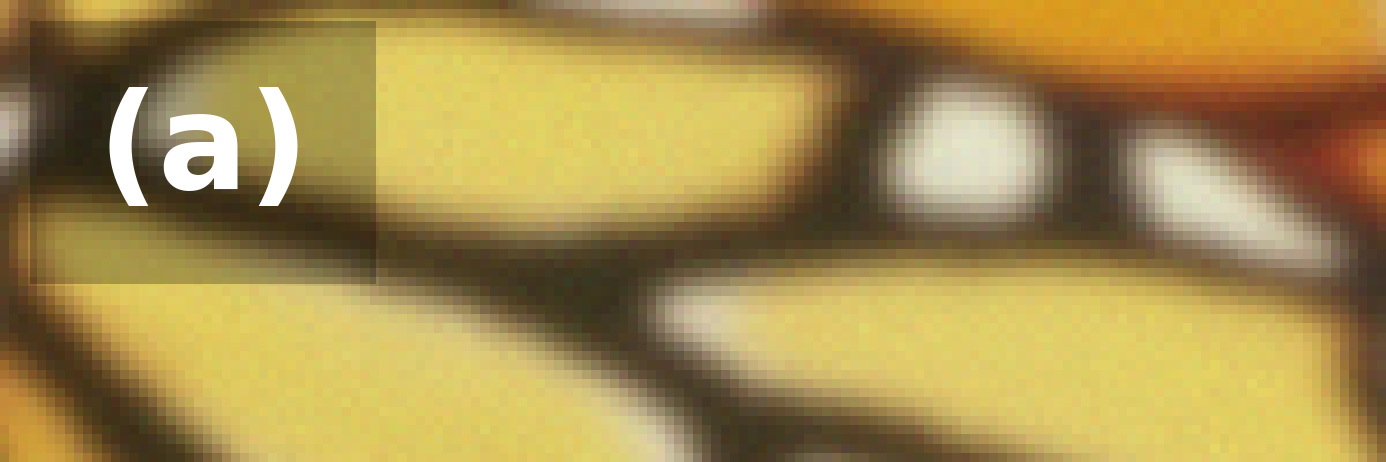} \\[2pt]
        \includegraphics[width=\linewidth]{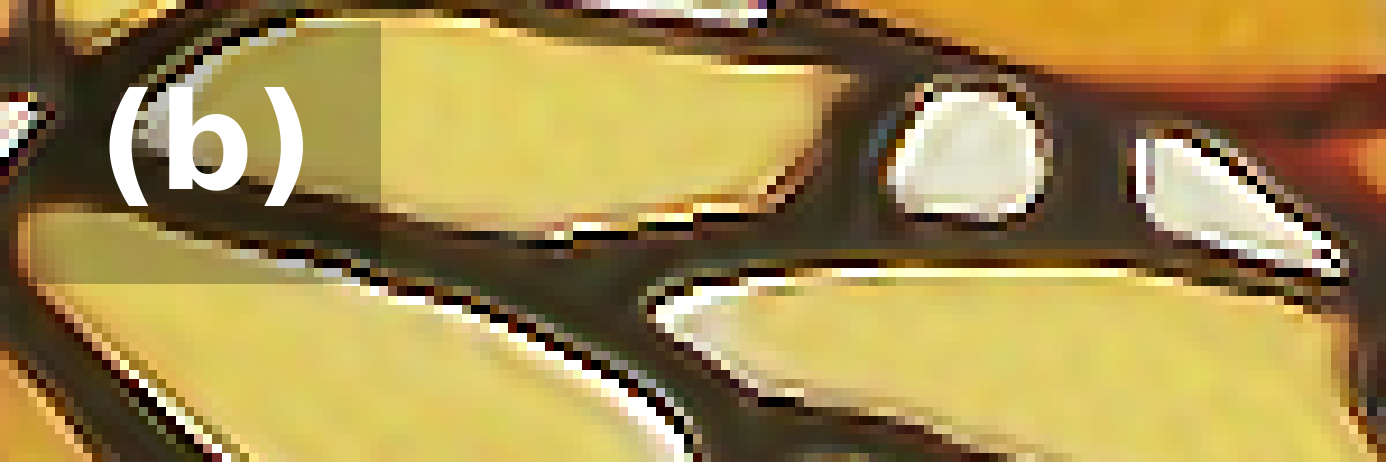} \\[2pt]
        \includegraphics[width=\linewidth]{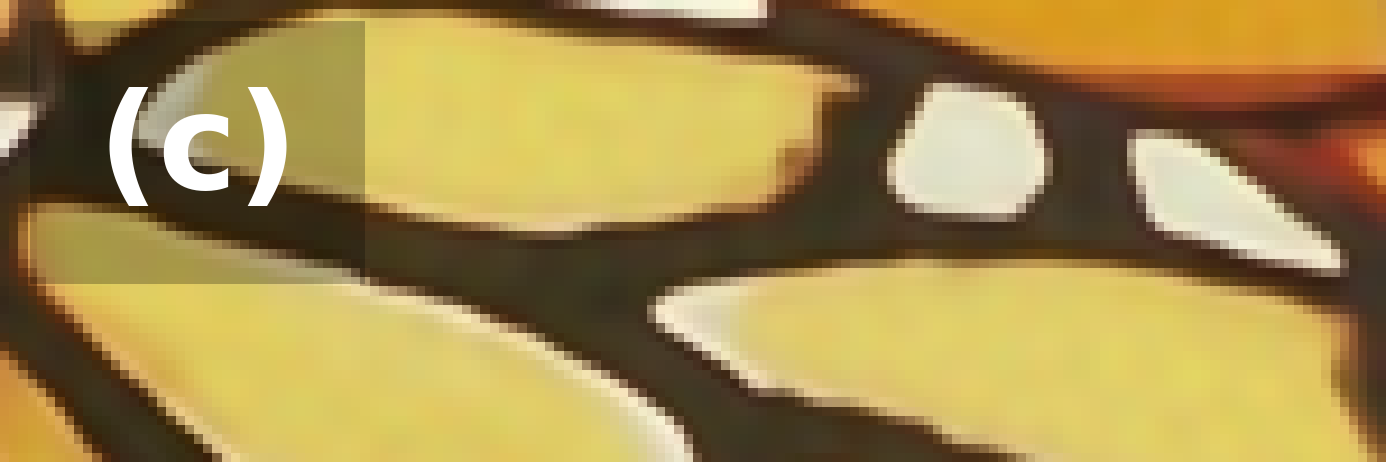} \\[2pt]
        \includegraphics[width=\linewidth]{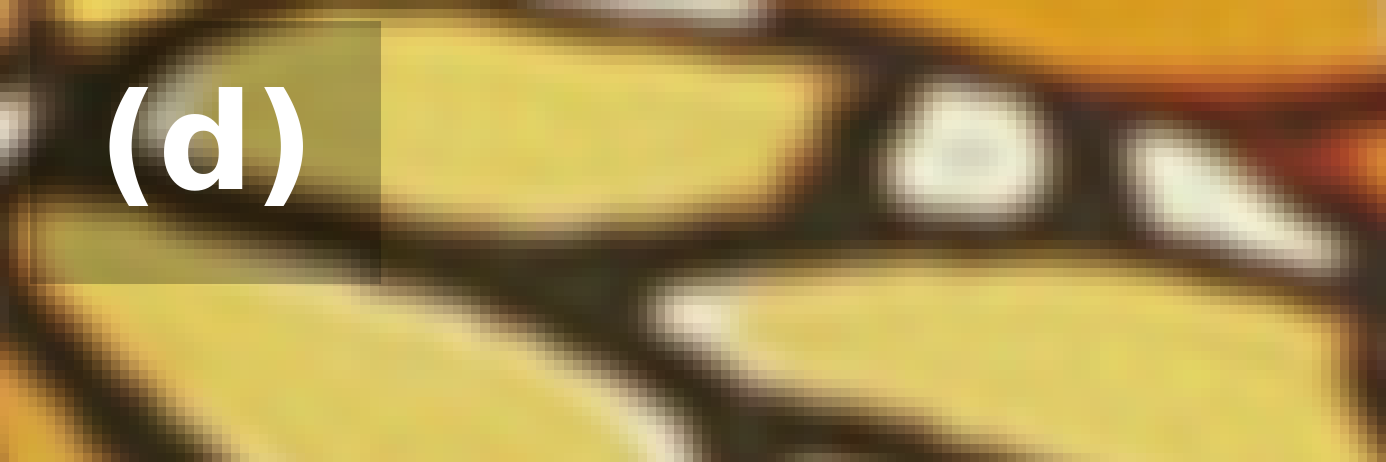}
    \end{minipage}%
    \hfill
    \begin{minipage}{0.58\linewidth}
        \centering
        \includegraphics[width=\linewidth]{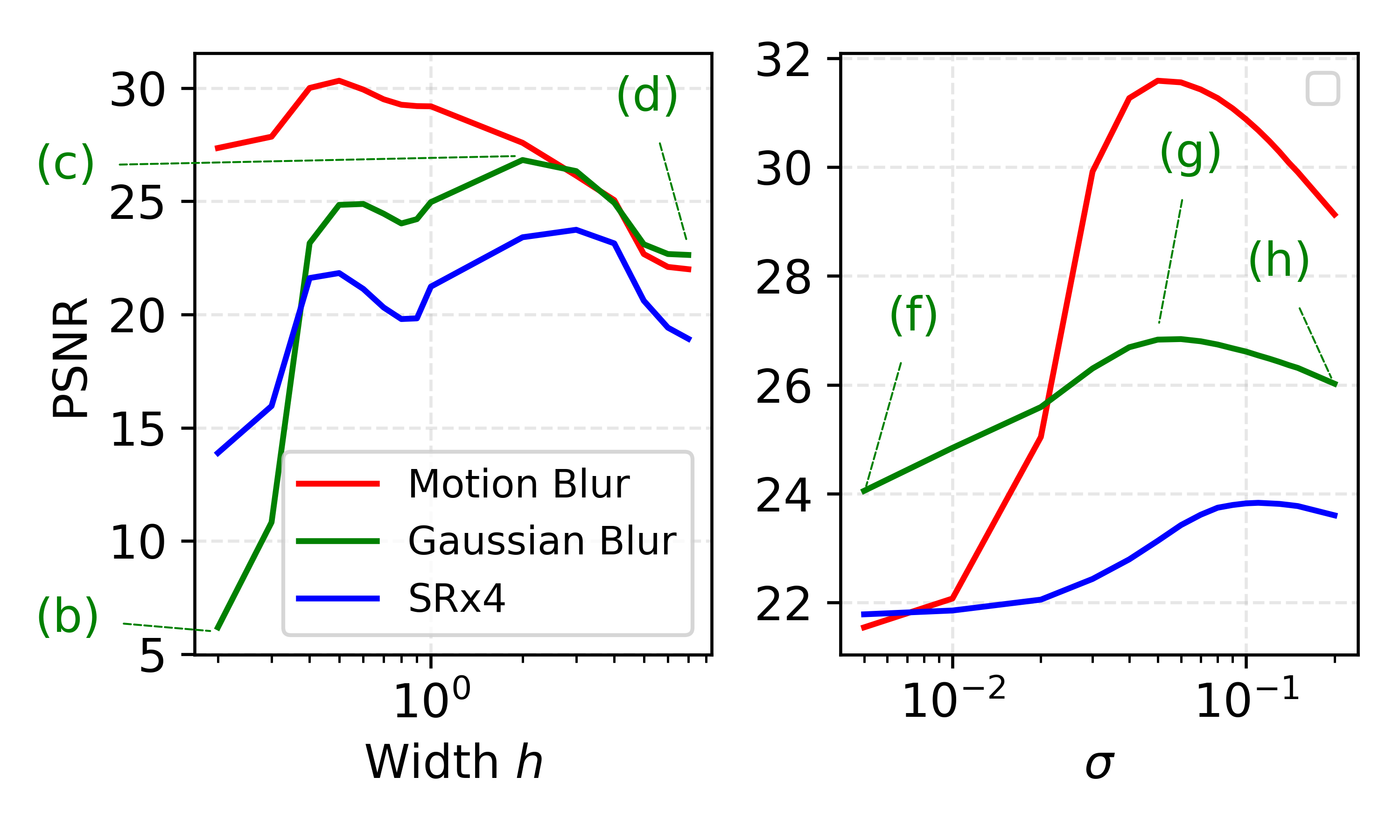}
    \end{minipage}%
    \hfill
    \begin{minipage}{0.18\linewidth}
        \centering
        \vspace{-1em}
        \includegraphics[width=\linewidth]{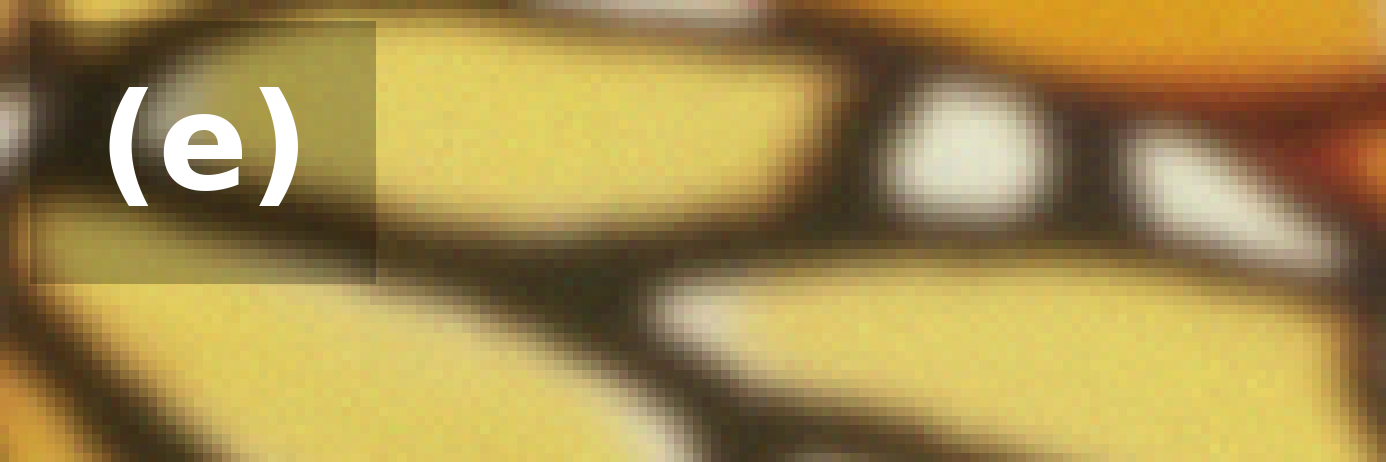} \\[2pt]
        \includegraphics[width=\linewidth]{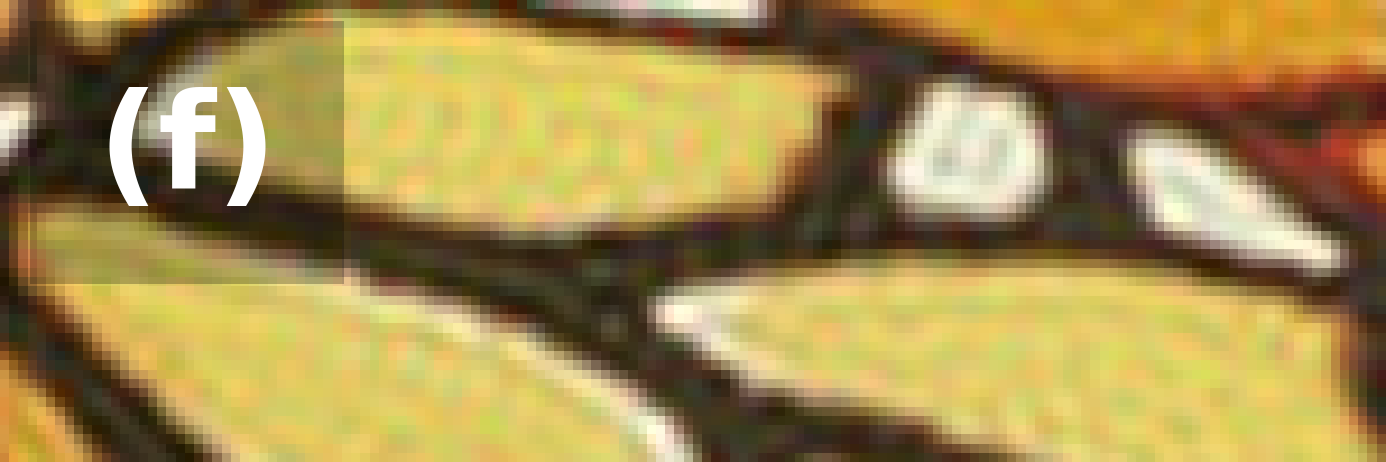} \\[2pt]
        \includegraphics[width=\linewidth]{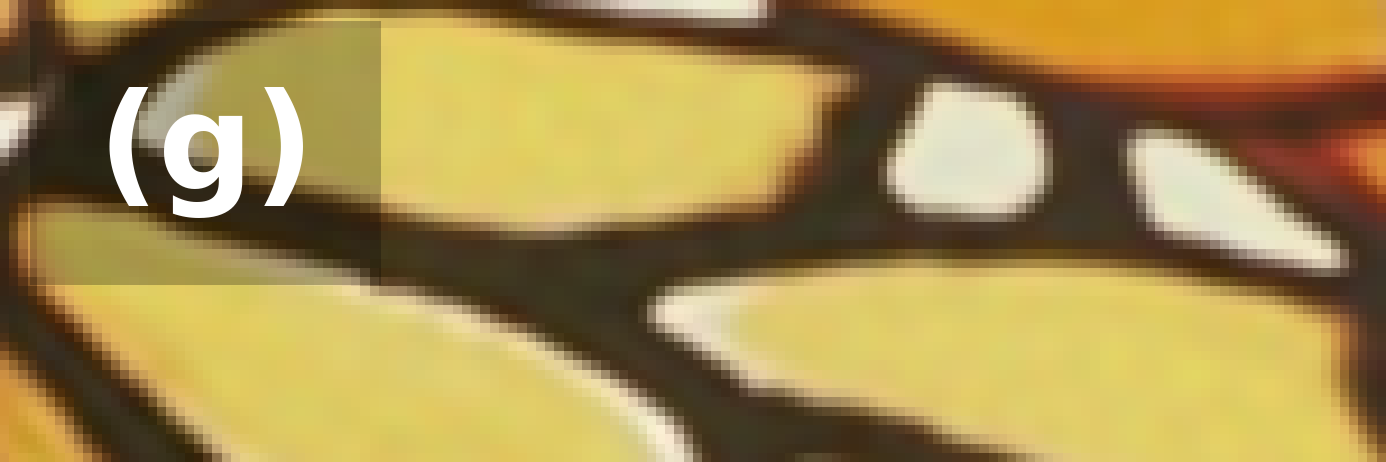} \\[2pt]
        \includegraphics[width=\linewidth]{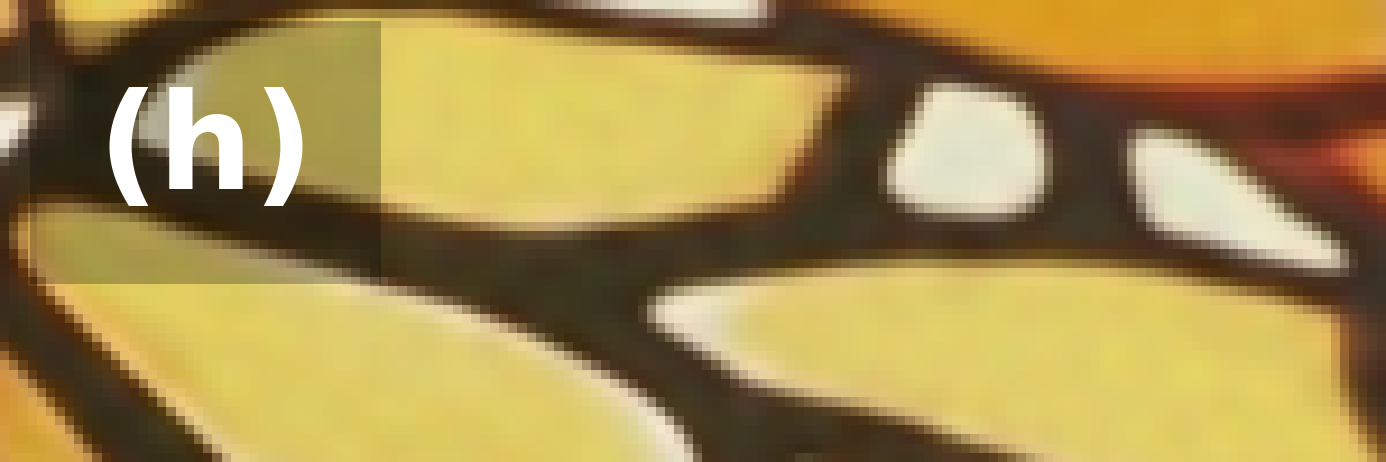}
    \end{minipage}%
    \vspace{-1em}
    \caption{\small Influence of the degradation operator on the FiRe restoration algorithm on the Set3C algorithm for different problems. Left plot: FiRe with Restormer Gaussian prior. We run the algorithm with different kernel widths. Reconstructions (details) corresponding to points (b), (c), (d) on the plot are shown on the left. Right plot: same, but with a DRUNet denoiser restoration prior; (f), (g) and (h) panels show the associated reconstructions. (a) and (e) show the degraded measurements $y$.}
    \label{fig:ablation_study}
\end{figure}

\end{document}